\documentclass[journal]{IEEEtran}
\usepackage{amsmath,amssymb,amsthm}
\usepackage{booktabs}
\usepackage{graphicx}
\usepackage{cite}
\usepackage{url}
\usepackage{placeins}

\newtheorem{theorem}{Theorem}
\newtheorem{lemma}{Lemma}
\newtheorem{corollary}{Corollary}
\newtheorem{proposition}{Proposition}
\theoremstyle{definition}

\theoremstyle{remark}

\newcommand{\F}{\mathbb{F}_2}
\newcommand{\kc}{k}

\newcommand{\bench}[1]{\texttt{\benchsplit#1\relax}}
\def\benchsplit#1-#2\relax{#1-\discretionary{}{}{}#2}

\newif\ifshowmat
\showmattrue

\title{A Version Space Approach for Digital Circuit Analysis}

\author{Mitchell~A.~Thornton,~\IEEEmembership{Senior Member,~IEEE}
\thanks{M. A. Thornton is with the Darwin Deason Institute for Cyber
Security and the Department of Electrical and Computer Engineering,
Southern Methodist University, Dallas, TX 75275 USA (e-mail:
mitch@smu.edu, ORCID \mbox{0000-0003-3559-9511}).}}

\begin{document}
\raggedbottom
\maketitle

\begin{abstract}
Many questions about a digital circuit take the same form.  A hidden object is
consistent with a set of observations, and one wants to know how many objects
remain consistent and which observation to make next.  The set of surviving
candidates is the \emph{version space}, and its size, reported on a logarithmic
scale, is a direct measure of how much the observations have settled.  This
paper develops the version-space view as a single method and applies it to two
problems in circuit analysis that are usually treated as unrelated.  The first
is probabilistic combinational equivalence checking, where the candidates are
Boolean functions and the observations are modified-Haar spectral
coefficients.  A method proposed in 2002 posed exactly this counting problem
and solved only two special cases, leaving the general case as an enumeration
that grew exponentially in the number of observations.  We close it: a
reparameterization onto block sums turns the dependence among nested
coefficients into locality, a sum--product recursion then counts the surviving
functions exactly in time polynomial in the truth-table size, closed forms
follow for a single coefficient, for a coefficient pair, and for every
ancestor-closed set, and the error of the independence approximation that the
2002 work resorted to is shown to equal a computable lattice index.  Every
formula is checked against exhaustive enumeration and reproduces the tables of
the 2002 paper.  The second application is key counting for logic-locked
netlists, where the candidates are keys and the observations are oracle
responses.  The same recursion, run over the gate-level factor graph, computes
the number of keys still consistent with a set of queries; across seventy
instances of the TrustHub obfuscation release the surviving entropy falls
below the advertised key length every time.  The two applications are one
method: a witness supplies observations, each observation removes candidates,
and the version space is counted exactly.
\end{abstract}


\section{Introduction}
\label{sec:intro}

\IEEEPARstart{A}{} recurring pattern in digital circuit analysis is the
following.  Some object is hidden, a reference function, a secret key, a fault,
a hypothesis about structure.  A witness supplies observations, and each
observation is consistent with some of the candidate objects and inconsistent
with others.  As observations accumulate, the set of candidates that remain
consistent with all of them shrinks.  Two questions govern how the analysis
proceeds: how many candidates still remain, and which observation should be
requested next.

The set of candidates consistent with the observations seen so far is called
the \emph{version space}, a term from computational learning
theory~\cite{mitchell1982}, and its size is the standard measure of how far a
body of evidence has narrowed a hypothesis class~\cite{haussler1988}.  This
paper takes the version space as the organizing idea for a class of circuit
analyses and shows that computing its size exactly, rather than estimating it,
is both possible and useful in problems where it has previously been thought
out of reach.

\subsection{The Version Space as a Measure}
\label{sec:vsmethod}

Fix a finite set $\mathcal{H}$ of candidate objects and a stream of
observations $o_1, o_2, \dots$, each of which either agrees or disagrees with a
given candidate.  After $t$ observations the version space is
\[
V_t \;=\; \{\, h \in \mathcal{H} \;:\; h \text{ agrees with } o_1,\dots,o_t \,\},
\]
a nested, shrinking family $V_0 \supseteq V_1 \supseteq \cdots$.  The paper
reports $\log_2 |V_t|$ rather than $|V_t|$, for a reason that is the same in
both applications below.  Before any observation the witness holds
$|V_0|=|\mathcal{H}|$ candidates and knows nothing to separate them, so the
honest description of that state is the uniform distribution over
$\mathcal{H}$, whose Shannon entropy is $\log_2|\mathcal{H}|$.  Because every
surviving candidate remains perfectly consistent with everything observed, the
posterior after $t$ observations is uniform on $V_t$ and its entropy is
$\log_2|V_t|$.  The two figures are therefore directly comparable, and their
difference $\log_2|\mathcal{H}| - \log_2|V_t|$ is exactly the number of bits
the observations have supplied.  Two consequences are used throughout: because
the posterior is uniform, an exact count is an exact entropy with no
estimation step in between; and because the version space shrinks
monotonically, the sequence of counts records the rate at which evidence
accumulates rather than a single end-of-analysis number.

Two computational questions attach to the version space.  The first is
\emph{counting}: evaluate $|V_t|$.  Naively this is model counting, which is
$\#\mathrm{P}$-hard in general~\cite{valiant1979}, so the useful cases are
those in which the structure of the observations makes the count tractable.
The second is \emph{ordering}: given the candidates that survive, decide which
observation to request next.  Under the uniform posterior the expected
one-step reduction in $\log_2|V_t|$ from an observation is the Shannon entropy
of that observation's outcome over the survivors, so requesting the
maximum-entropy available observation is the greedy policy for expected
information gain.  This criterion is classical Bayesian experimental
design~\cite{lindley1956,chaloner1995}; what the applications below add is that
the posterior it needs is available exactly, in closed form, rather than by
estimation.

\subsection{Two Applications}
\label{sec:twoapps}

The paper develops two instances of this method that are usually studied
independently.

The first is \textbf{probabilistic combinational equivalence checking (CEC)}.
Here the hidden object is a Boolean function, the candidates are all functions
of $n$ variables, and the observations are modified-Haar spectral coefficients:
two circuits are compared by computing a few coefficients of each, and
agreement on those coefficients is evidence, never proof, of equivalence.  The
strength of that evidence is precisely the number of functions still consistent
with the coefficients observed, the size of the version space.  A method
proposed in 2002 by Thornton, Drechsler and G\"unther posed this counting
problem and solved it only for a single coefficient and for coefficients on
disjoint supports, leaving the general case as an enumeration that grew
exponentially in the number of coefficients~\cite{thornton2002vlsi}.
Section~\ref{sec:cec} closes the general case.

The second is \textbf{key counting for logic-locked netlists}.  Here the hidden
object is a secret key, the candidates are all $2^k$ key values, and the
observations are input--output pairs read from an unlocked oracle chip: each
query rules out the keys inconsistent with the observed response.  The
surviving key set is the version space, and its logarithm is directly
comparable to the advertised key length.  Section~\ref{sec:lockapp} shows that
the same counting recursion, run over the gate-level factor graph of the
netlist, computes this quantity exactly, and reports it across the published
Trust-Hub obfuscation benchmarks.

The two applications share more than a vocabulary.  In both, the observations
are linear or gate-level constraints whose joint structure, exploited
correctly, makes an apparently intractable count tractable; in both, the exact
posterior turns query selection from a heuristic into a computation; and in
both, the naive independence approximation that a practitioner reaches for when
an exact count seems unavailable is wrong in a direction that inflates
confidence.  The CEC application is presented first because it is where the
machinery originated and where every quantity has a closed form against which
the method can be checked exactly.  Since combinational equivalence checking
has since been settled in practice by satisfiability sweeping and by computer
algebra, the CEC results here are offered not as a competitive checker but as
the completion of the counting problem that the 2002 method left open, and as
the exactly-solvable setting that validates the machinery the second
application then carries into hardware security.

Section~\ref{sec:haar} sets up the modified-Haar transform on switching
circuits and its lineage.  Section~\ref{sec:cec} is the first application, the
CEC counting problem and its solution.  Section~\ref{sec:lockapp} is the
second application, key counting for locked netlists, with the experimental
program on the published benchmarks.  Section~\ref{sec:disc} discusses the
limits of both.

\section{The Haar Transform Applied to Switching Circuits}
\label{sec:haar}

Both applications in this paper are counting problems over constraints that are
rows of a spectral transform matrix.  In the equivalence-checking application
the constraints are literally modified-Haar coefficients; in the locking
application the gate-level constraints are handled by the same block-structured
recursion that the Haar system makes possible.  This section fixes the
transform, its notation, and the lineage of its use in switching theory, so
that the counting sections can refer to it directly.

\subsection{Spectral Transforms of Switching Functions}

A completely specified switching function $f:\{0,1\}^n\rightarrow\{0,1\}$ is
represented for spectral purposes by its $\{+1,-1\}$-encoded truth-table
vector $\hat s$ of length $N=2^n$, with logic~$0$ mapped to $+1$ and logic~$1$
to $-1$.  A spectral transform is a fixed integer matrix $\mathbf{T}$ of order
$N$; the spectrum of $f$ is $\mathbf{T}\hat s$, and each row of $\mathbf{T}$
defines one spectral coefficient as an integer linear functional of the truth
table.  The Walsh--Hadamard transform, whose rows are the characters of
$\mathbb{Z}_2^n$, is the most widely used such transform in switching
theory~\cite{karpovsky1976,hurst1985,karpovsky1985}, and the Reed--Muller
transform plays the analogous role over $\mathrm{GF}(2)$.  The Haar transform
is distinguished among them by being \emph{local}: most of its rows are
supported on a single dyadic block of the truth table rather than on all of
it, so a Haar coefficient reports on a localized region of the function.  This
locality is what makes the counting problem of Section~\ref{sec:cec}
tractable, and it is the reason the Haar transform, rather than Walsh or
Reed--Muller, is the natural setting here.

\subsection{The Modified Haar Transform}

We use the \emph{modified} Haar transform, in which every nonzero entry is
$+1$ or $-1$ and the rows are indexed by the nodes of a binary
\emph{block tree} over the truth table~\cite{thornton1997icics,thornton1995tcad}.
The root row $H_0$ is the all-ones vector, the global signed sum of the truth
table.  Each remaining row $H(j,c)$ is attached to a node of the tree at depth
$j$ and dyadic position $c$: it is $+1$ on the left half of that node's block,
$-1$ on the right half, and $0$ elsewhere.  A node at depth $j$ has a block of
size $N/2^{j}$, so rows deeper in the tree are supported on smaller blocks and
report on finer detail.  Writing the row as an operator on Shannon cofactors,
$H(j,c)$ is the signed sum of the cofactor of $f$ selected by the path from the
root to node $(j,c)$, which is why the modified Haar spectrum is naturally read
as a tree of cofactor sums rather than as a flat vector.

For $n=3$ the modified Haar transformation matrix, with each row labeled by the
cofactor functional it computes, is
\begin{equation}
\label{eq:haar3}
\mathbf{T}_3=
\begin{array}{r@{\;}c}
f & \begin{bmatrix} 1&1&1&1&1&1&1&1 \end{bmatrix}\\[1pt]
x_1 & \begin{bmatrix} 1&1&1&1&\bar 1&\bar 1&\bar 1&\bar 1 \end{bmatrix}\\[1pt]
x_2 f_{\bar x_1} & \begin{bmatrix} 1&1&\bar 1&\bar 1&0&0&0&0 \end{bmatrix}\\[1pt]
x_2 f_{x_1} & \begin{bmatrix} 0&0&0&0&1&1&\bar 1&\bar 1 \end{bmatrix}\\[1pt]
x_3 f_{\bar x_1\bar x_2} & \begin{bmatrix} 1&\bar 1&0&0&0&0&0&0 \end{bmatrix}\\[1pt]
x_3 f_{\bar x_1 x_2} & \begin{bmatrix} 0&0&1&\bar 1&0&0&0&0 \end{bmatrix}\\[1pt]
x_3 f_{x_1\bar x_2} & \begin{bmatrix} 0&0&0&0&1&\bar 1&0&0 \end{bmatrix}\\[1pt]
x_3 f_{x_1 x_2} & \begin{bmatrix} 0&0&0&0&0&0&1&\bar 1 \end{bmatrix}
\end{array}
\end{equation}
where $\bar 1$ denotes $-1$.  The eight rows are the root sum, one half-block
difference at depth one, two quarter-block differences at depth two, and four
single-cell differences at depth three; the labels record that the higher-order
rows read progressively deeper Shannon cofactors of $f$.  The $0$ entries mark
the cells outside a row's block and are what make the transform local.  A row
at depth $j$ is supported on $2^{\,n-j}$ cells, so the maximum magnitude a
coefficient can attain decreases with depth, and the whole matrix is generated
by one recursive rule: split each block into its two dyadic halves and attach a
$(+1,-1)$ difference row to the node.

\subsection{Lineage and Past Applications}

Spectral methods for switching functions were placed on a systematic footing by
Karpovsky~\cite{karpovsky1976,karpovsky1985} and by Hurst, Miller and
Muzio~\cite{hurst1985}, who developed the Walsh, Reed--Muller and Haar spectra
as tools for synthesis, testing and classification.  Hurst identified the Haar
transform specifically as a vehicle for digital network
synthesis~\cite{hurst1981}, and Stankovi\'c and Falkowski later surveyed the
status of the Haar wavelet transform in this setting~\cite{stankovic2003}.  The
practical obstacle throughout was cost: a direct transform is an
$N\times N$ matrix--vector product, prohibitive for functions of many
variables.

Two developments removed the obstacle for the coefficients actually needed.
The first was the observation that the Haar spectrum can be computed directly
from a decision-diagram representation of the function without ever forming the
truth table: Falkowski and Chang gave forward and inverse transformations
between the Haar spectrum and binary decision
diagrams~\cite{falkowski1994}, Hansen and Sekine gave decision-diagram
techniques for the Haar wavelet transform~\cite{hansen1997}, and multi-terminal
and related diagrams~\cite{clarke1997mtbdd} supplied the arithmetic
substrate.  The second was the reformulation of individual spectral
coefficients as circuit \emph{output probabilities}.  A modified-Haar
coefficient equals, up to scale and offset, the probability that a cofactor of
$f$ evaluates to a given logic value, so it can be read off a decision diagram
by a single depth-first traversal that accumulates path probabilities in
$O(\text{diagram size})$ time~\cite{thornton1997icics,thornton1995tcad}.  This
output-probability route to the spectrum, in which each coefficient is obtained
from an inexpensive probability computation over a partial diagram rather than
from a matrix product, is the computational basis for the equivalence-checking
method summarized in the next section, and the general treatment of spectral
techniques in VLSI CAD is collected in~\cite{thornton2001book}.

The point to carry forward is structural.  Because each modified-Haar row is a
signed difference on a dyadic block, the coefficients are organized by the
block tree, and conditioning on a coefficient constrains the signed sum of one
block.  Section~\ref{sec:cec} shows that this block-tree organization is
exactly what converts the apparently intractable problem of counting the
functions consistent with a set of observed coefficients into a local,
tree-structured recursion.

\section{Application I: Combinational Equivalence Checking}
\label{sec:cec}

The first application is the one in which the version-space machinery
originated, and in which every quantity has a closed form.  The candidates are
Boolean functions of $n$ variables, the observations are modified-Haar
coefficients, and the version space is the set of functions consistent with the
coefficients observed.

\subsection{The Previous CEC Method}
\label{sec:cec2002}

Probabilistic combinational equivalence checking compares two circuits without
proving them equal.  A few modified-Haar coefficients of each are computed from
its decision diagram by the output-probability route of
Section~\ref{sec:haar}, and if the coefficients agree the circuits are declared
equivalent with a quantified confidence.  Thornton, Drechsler and G\"unther
introduced this method in 2002 and identified the right measure of
confidence~\cite{thornton2002vlsi}: if $\kc(C)$ is the number of Boolean
functions whose coefficients agree with the observed set $C$, then a random
function consistent with $C$ is the intended one with probability $1/\kc(C)$,
so $\kc(C)$, the size of the version space, is exactly the strength of the
evidence.  Each additional matching coefficient shrinks $\kc$ and sharpens the
conclusion; a single mismatch proves inequivalence outright.

The 2002 work computed $\kc(C)$ in only two cases.  For a single coefficient it
gave a binomial closed form, and for coefficients on disjoint supports it
multiplied the individual counts under an independence argument.  The general
case, in which the conditioned coefficients share supports and are therefore
statistically dependent, was left as an enumeration over coefficient-value
grids that grew exponentially in the number of coefficients (Eq.~(41) there),
and the dependence was attacked in coefficient space, where it has no product
structure.  That is the gap this section closes.  The remainder of the section
develops the exact general recursion, the closed forms that the two 2002 cases
turn out to be extremes of, the exact error of the independence approximation
the 2002 method fell back on, and machine verification against the tables of
the 2002 paper.

\subsection{The Counting Problem, Reparameterized}
\label{sec:count}

Let $C$ denote a set of conditioned coefficients with observed values,
formally a partial map from $\{H_0\}\cup\{H(j,c)\}$ into $\mathbb{Z}$,
and define
\[
\kc(C) \;=\; \#\bigl\{\, g:\F^n\rightarrow\F \;\bigm|\;
g \vDash C \,\bigr\},
\]
where $g \vDash C$ abbreviates the statement that every coefficient of $g$
named in $C$ equals its observed value.  The governing expression of the 2002
method is
\begin{equation}
\Pr\bigl[\, f\equiv g \bigm| \text{the coefficients in } C
\text{ match} \,\bigr] \;=\; \frac{1}{\kc(C)},
\label{eq:posterior}
\end{equation}
so everything reduces to computing $\kc$.  The difficulty in 2002 was that
nested coefficients are statistically dependent, and the dependence was
attacked in coefficient space, where it has no product structure.  The way
out is to stop conditioning on coefficients and to start conditioning on the
block sums they constrain.

\begin{lemma}[Local bijection]
\label{lem:bij}
At every internal node of the block tree, the map
$(S,H)\mapsto(S_L,S_R)$ with $S_L=\tfrac{S+H}{2}$,
$S_R=\tfrac{S-H}{2}$ is a bijection between parent-sum/coefficient pairs
of equal parity and child-sum pairs, with inverse $S=S_L+S_R$,
$H=S_L-S_R$.  Consequently the full modified Haar spectrum determines
all block sums and conversely, and the transform is invertible.
\end{lemma}

In plain terms, a Haar coefficient at a node is an instruction to split that
node's total in a particular way between its two children.  Viewing the
spectrum as a tree of splits rather than as a vector of values converts the
dependence structure that obstructed the 2002 analysis into locality, and
locality is what probability theory rewards.

\begin{lemma}[Free blocks]
\label{lem:free}
Under the uniform prior, a block of size $m$ containing no conditioned
coefficient in its subtree admits exactly
$\binom{m}{(m+S)/2}$
fillings with block sum $S$ (zero unless $|S|\le m$ and
$S\equiv m \bmod 2$).
\end{lemma}

\subsection{The General Recursion}
\label{sec:main}

The first theorem computes $\kc(C)$ for an arbitrary conditioned set, with no
independence or structure assumptions, in time polynomial in the truth-table
size.  The intuition is one sentence.  Walk the tree from the leaves up,
carrying at each block the exact census of fillings by block sum,
multiplying censuses where a conditioned coefficient forces the split and
convolving them where it does not.

\begin{theorem}[Tree recursion for arbitrary coefficient sets]
\label{thm:dp}
For any conditioned set $C$, define for each block $B(j,c)$ the vector
$g_{(j,c)}(S)$ = the number of fillings of the block with sum $S$
consistent with all conditioned coefficients inside its subtree.  Writing
$g_L = g_{(j+1,c0)}$ and $g_R = g_{(j+1,c1)}$ for the two child censuses,
\[
g_{(j,c)}(S) =
\begin{cases}
[\,S=\pm1\,] & j=n{+}1,\\[3pt]
g_L\!\bigl(\tfrac{S+h}{2}\bigr)\,g_R\!\bigl(\tfrac{S-h}{2}\bigr)
& H(j,c)=h \in C,\\[3pt]
\displaystyle\sum_{S_L+S_R=S} g_L(S_L)\, g_R(S_R)
& H(j,c)\notin C,
\end{cases}
\]
with the middle case zero on any parity failure, and
$\kc(C) = g_{(1,())}(h_0)$ when $H_0=h_0\in C$, else
$\sum_S g_{(1,())}(S)$.  Blocks disjoint from the conditioned set
contribute Lemma~\ref{lem:free} in closed form, so the work is confined
to the union of root paths of conditioned nodes.  Performing each
unconstrained step entry by entry costs $O(4^n)=O(N^2)$ integer
operations.  Section~\ref{sec:cost} improves this to $O(N\log^2 N)$.
\end{theorem}

\subsection{Recursion Details}
\label{sec:whatitis}

Theorem~\ref{thm:dp} is the sum--product algorithm over the counting
semiring, on a tree-structured factor graph.  Take as variables the
truth-table entries $s_i\in\{\pm1\}$ together with a partial sum $S(v)$ at
each node $v$ of the block tree, and as factors
$\delta[\,S(v)=S(v_L)+S(v_R)\,]$ at every internal node and
$\delta[\,S(v_L)-S(v_R)=h_v\,]$ at every conditioned node.  Every factor has
scope $\{S(v),S(v_L),S(v_R)\}$, so the factor graph \emph{is} the block tree
and is cycle-free by construction.  The forced split of the middle case and
the convolution of the third are the two message operations.  No new
inference machinery is introduced here.  The generalized distributive
law~\cite{ajimceliece2000} and the factor-graph
formulation~\cite{kschischang2001} cover it, and the same
census-by-residual-margin recursion has been in production statistical
software since the network algorithm for Fisher's exact
test~\cite{mehtapatel1983}.  Recent work in side-channel analysis makes the
same move of replacing an approximate belief propagation with exact inference
on a compiled tractable representation~\cite{wedenig2024}.

The contribution is therefore not the recursion.  It is the observation that
the Haar block-sum system satisfies the conditions under which the recursion
is cheap, together with what that exactness then buys.  The conditions are
isolated below because only the first is usually named.

\begin{proposition}[Scope]
\label{prop:scope}
Let $L_1,\dots,L_q$ be integer linear functionals on $\{\pm1\}^N$ and let
$\sigma_v$ denote the vector of pending partial sums at node $v$ of the
tree of constraint supports.  Sum--product over that tree computes the
exact solution count in $O\!\left(\sum_v |\mathrm{range}(\sigma_v)|^2\right)$
operations.  Three conditions make the state space small:
\begin{enumerate}
\item[(C1)] \emph{Laminar supports.} Every pair of supports is nested or
disjoint, so the tree exists.
\item[(C2)] \emph{Constancy on descendants.} For every $L_j$ and every
laminar node $v\subsetneq\mathrm{supp}(L_j)$, the coefficient vector of
$L_j$ is constant on $v$.
\item[(C3)] \emph{Bounded coefficients.} Coefficients are $\pm1$.
\end{enumerate}
Under (C1)--(C3) the state at $v$ is one integer of range $|v|+1$.
\end{proposition}

Condition (C2) is the load-bearing one, and it is why the Haar system in
particular is cheap.  A node $v$ at depth $d$ has up to $d$ conditioned
ancestors, so $\sigma_v$ is a priori $d$-dimensional and the state space is
$N^{O(d)}$.  It collapses to one dimension because $v$ lies entirely inside a
single dyadic half of each of its ancestors, so each ancestor's coefficient
vector is constant on $v$ and every pending sum is $\pm\sum_{i\in v}s_i$, the
same number up to sign.  Replace the Haar sign pattern at each node by an
arbitrary $\pm1$ pattern, keeping the dyadic supports and hence (C1), and the
collapse fails.

None of the three conditions is implied by the others, and (C1) alone has no
computational content.  Let the laminar family be the $N$ singletons together
with $[N]$, and impose the single constraint $\sum_i a_i x_i = h$ with
$x\in\{0,1\}^N$.  Encoding a bipartite perfect-matching instance into the
coefficients in a base larger than any achievable row value makes the
solution count equal to the permanent of a 0/1 matrix, with coefficients of
$O(N\log N)$ bits.  Counting is therefore
$\#\mathrm{P}$-complete~\cite{valiant1979}.  A second family, $2r+1$
constraints of common support with coefficients in $\{1,2\}$, gives the same
count with bounded coefficients, so (C3) does not rescue (C1) either.  Both
were verified against brute force (Section~\ref{sec:verify}).

It is equally worth recording what does \emph{not} explain the tractability.
The root coefficient $H_0$ has support $[N]$, so the primal graph of the
constraint hypergraph is a clique on all $N$ truth-table entries and the
system has treewidth $N-1$.  The tree exploited here is not a tree
decomposition of that hypergraph.  It is a junction tree over an extended
formulation carrying the auxiliary variables $S(v)$.  Bounded treewidth is
also not necessary in general.  Sliding-window sums of width $w$ are
maximally non-laminar, no two supports being nested or disjoint, and are
counted in $O(N2^{w-1})$.  The two conditions are incomparable, and the
dichotomy theorems for counting constraint satisfaction, which do give a
bounded-treewidth characterization~\cite{dalmau2004}, assume bounded arity
and explicitly tabulated relations.  Neither hypothesis holds for a
succinctly described functional of arity $N$.

\subsection{Cost}
\label{sec:cost}

The unconstrained step is a convolution of two child censuses of length
$m/2+1$ at a node of block size $m$.  Entry by entry that is $O(m^2)$, and
since level $m$ holds $N/m$ nodes the total is $O(N^2)=O(4^n)$.
Sub-quadratic convolution reduces this to $O(N\log^2 N)$ arithmetic
operations on integers of at most $N$ bits, or $\tilde O(N^2)$ bit operations
in place of $\tilde O(N^3)$.  Conditioned nodes are cheaper still, the split
being forced.

The implementation realizes this by packing each census into a single integer
with a field width wide enough that no product coefficient carries,
multiplying once, and unpacking.  The arithmetic is exact throughout.
Measured against the entry-by-entry path on the unconstrained census, and
agreeing with it in every digit, the speedup grows with $n$ as the analysis
requires.  It is $1.70\times$ at $n=8$, $2.28\times$ at $n=12$ and
$3.55\times$ at $n=13$, where the naive path takes 176~s against 49~s.  The
measured naive time at $n=13$ exceeds what a pure $O(N^2)$ arithmetic model
predicts, which is the bit-complexity term appearing where the analysis
places it.

The point to take from Theorem~\ref{thm:dp} is not the formula and not the
algorithm.  It is the complexity class relative to what preceded it.  The
2002 update (Eq.~(41) there) enumerates coefficient-value grids and grows
exponentially in $|C|$, while the recursion above is polynomial in $N$ and
does not depend on $|C|$.  That independence is not itself a result, since
further constraints only prune states.  What matters is the disappearance of
the grid enumeration.  The problem resisted a general solution because the
dependence between nested coefficients has no product structure in
coefficient space.  In block-sum space it is a forced split at one tree node,
and is therefore local.

\subsection{Closed Forms}

\begin{theorem}[Vandermonde collapse, single coefficient]
\label{thm:single}
For a single conditioned coefficient with block size $m$ and value $h$,
$\kc = 2^{\,N-m}\binom{m}{(m+h)/2}$.  For $H_0=h$,
$\kc=\binom{N}{(N+h)/2}$.
\end{theorem}

\begin{corollary}
Theorem~\ref{thm:single} is Eq.~(24) of~\cite{thornton2002vlsi},
including its worked example $\kc(H_0{=}2)=4$ at $n=2$.
\end{corollary}

\begin{theorem}[The pair $(H_0,H_1)$, closed]
\label{thm:pair}
With $S_L=\tfrac{H_0+H_1}{2}$ and $S_R=\tfrac{H_0-H_1}{2}$,
\[
\kc(H_0,H_1) \;=\; \binom{N/2}{\tfrac{N/2+S_L}{2}}
\binom{N/2}{\tfrac{N/2+S_R}{2}},
\]
zero on any parity or range failure.
\end{theorem}

Theorem~\ref{thm:pair} gives in two binomial coefficients the joint value
that the 2002 method~\cite{thornton2002vlsi} could reach only through a
unit-step enumeration.  $H_0$ and $H_1$ are the sum and difference of the two
half-space sums.  Conditioning on both forces each half exactly, and each
half is then a free block.

\begin{theorem}[Ancestor-closed sets: product form]
\label{thm:anc}
Suppose $H_0\in C$ and $C$ is ancestor-closed: for every conditioned
$H(j,c)$, every coefficient on the root path above it is also
conditioned.  Then all block sums on the constraint skeleton are forced
by iterating Lemma~\ref{lem:bij}, and
\[
\kc(C) \;=\; \prod_{B\in\mathcal{F}(C)} \binom{|B|}{\tfrac{|B|+S_B}{2}},
\]
where $\mathcal{F}(C)$ is the fringe of maximal blocks containing no
conditioned coefficient and $S_B$ the forced sum of $B$.
\end{theorem}

\begin{corollary}[Disjoint supports are independent]
\label{cor:indep}
If $C_A$ and $C_B$ condition coefficients supported in disjoint subtrees
(neither conditioning $H_0$), then
$2^N \kc(C_A\cup C_B) = \kc(C_A)\,\kc(C_B)$, so matching events on disjoint
supports are independent.  This recovers the independence used
throughout the experiments of~\cite{thornton2002vlsi}, including its
$\kc(H_2{=}0, H_3{=}{-}2)=2$ example.
\end{corollary}

\subsection{Error due to an Independence Assumption}
\label{sec:gap}

Corollary~\ref{cor:indep} says that coefficients on disjoint supports are
exactly independent.  Practitioners who need $\kc$ but lack a way to compute
it have taken the further step of assuming independence in general, and the
cost of that step has an exact answer.  Write
\[
\kc_{\mathrm{ind}}(C) \;=\; 2^{N}\prod_{i\in C}\Pr\bigl[\,i\text{-th
coefficient}=h_i\,\bigr]
\]
for the count implied by treating every constraint as independent, and let
$P_m(s)=\binom{m}{(m+s)/2}2^{-m}$ be the signed-sum distribution over $m$
cells, so that a coefficient with block size $m$ has
$\Pr[\,\cdot=h\,]=P_m(h)$.

Each coefficient is a linear functional on the $\pm1$-encoded truth table.
Write $\varepsilon_i\in\{0,+1,-1\}^{N}$ for the vector it reads, namely the
all-ones vector for $H_0$ and, for $H(j,c)$, the vector that is $+1$ on the
left half of the block, $-1$ on the right half, and $0$ elsewhere.  Stack
these as the rows of an integer matrix $M_C$, so that the conditioned values
are $M_C\,\hat s$.

\begin{proposition}[Independence gap]
\label{prop:gap}
Let $t=|C|$ and let $M_C$ be as above.  Under the local central limit
approximation $P_m(s)\approx\sqrt{2/(\pi m)}\,e^{-s^{2}/2m}$, as the block
sizes grow
\[
\frac{\kc(C)}{\kc_{\mathrm{ind}}(C)} \;\longrightarrow\;
\bigl[\,\mathbb{Z}^{t} : M_C(\mathbb{Z}^{N})\,\bigr],
\]
the index of the image lattice of the constraint map.  The index is a power
of two, is computable in polynomial time from the Smith normal form of $M_C$,
and does not depend on $n$.
\end{proposition}

\begin{proof}
The coefficient values attainable by integer vectors are exactly the lattice
$M_C(\mathbb{Z}^{N})\subseteq\mathbb{Z}^{t}$.  The independence estimate
treats the $t$ coordinates as ranging over $\mathbb{Z}^{t}$ independently, so
it spreads the same total mass over a superlattice of index
$[\mathbb{Z}^{t}:M_C(\mathbb{Z}^{N})]$.  On the attainable sublattice the true
density is larger by exactly that factor.  Formally, substituting the local
central limit form makes each side a Gaussian density on the relevant lattice,
the quadratic exponents agree because $\sum_i h_i^{2}/m_i$ is the squared norm
of the orthogonal projection of $\hat s$ onto $\mathrm{span}\{\varepsilon_i\}$
and is therefore basis-independent, and what remains is the ratio of the
lattice covolumes, which is the stated index.  The index is a power of two
because every $\varepsilon_i$ has entries in $\{0,\pm1\}$ and the Haar
functionals are linearly independent, so $M_C$ has full row rank and its
elementary divisors are powers of two.
\end{proof}

\begin{corollary}[Ancestor-closed sets]
\label{cor:gapanc}
If $C$ is ancestor-closed then
$\kc(C)/\kc_{\mathrm{ind}}(C)\rightarrow 2^{|C|-1}$.
\end{corollary}

\begin{proof}
Induct on the number of conditioned coefficients.  For $C=\{H_0\}$ the matrix
is a single all-ones row and the index is $1$.  Adding an ancestor-closed
coefficient splits one fringe block into two halves.  In the block-sum basis
this substitutes for one generator two whose sum is the old one, which
doubles the index.  Hence the index is $2^{|C|-1}$.
\end{proof}

\begin{corollary}[Disjoint supports]
\label{cor:gapdisj}
If the supports of the coefficients in $C$ are pairwise disjoint the index is
$1$, recovering Corollary~\ref{cor:indep}.
\end{corollary}

The lattice picture also explains why a natural guess fails.  One might
expect the gap to be governed by how much of each constrained block is pinned
from above, that is by the number of coefficients whose whole root path is
also constrained.  That is correct for ancestor-closed sets and wrong in
general.  A coefficient together with \emph{both} of its children is linearly
dependent in the same way, gives an index of $2$, and contains no root path
at all.  Constraints can pin a block sum laterally and from below, not only
from the root down.  The lattice index counts all of these at once.

The practical reading of Proposition~\ref{prop:gap} is a warning about
product approximations in general, and the key-space claim of
Section~\ref{sec:lockapp} is one.  An independence assumption is not mildly
optimistic on a coupled constraint set.  It is wrong by a factor equal to a
lattice index, that index grows with the amount of coupling, and it is wrong
in the direction that inflates confidence.  The proposition applies to
orthogonal integer linear functionals, which oracle responses are not, so it
does not give the size of the error for a locked netlist.  It gives the
reason to expect one, and Section~\ref{sec:expt} measures the error directly
instead.

The two exact results of the 2002 paper~\cite{thornton2002vlsi} are the
extreme special cases of Theorem~\ref{thm:anc}.  The first is a single
coefficient, for which the fringe is everything else.  The second is disjoint
supports, for which the skeletons never meet.  Everything in between, which
is where partial BDDs deposit their coefficients, is covered by the same
product form or, failing ancestor-closure, by Theorem~\ref{thm:dp}.

\subsection{Ordering the Coefficients}
\label{sec:seq}

\begin{theorem}[Exact posterior and refinement]
\label{thm:post}
Under the uniform prior, after observing that the coefficients in $C$
match the reference function's values,
$\Pr[E\mid C] = 1/\kc(C)$.  The posterior is nondecreasing as
coefficients are added, it equals $1$ when all $N$ coefficients are
conditioned, and a mismatch at any coefficient certifies inequivalence.
Each refinement re-runs only the root path of the new coefficient in
Theorem~\ref{thm:dp}.
\end{theorem}

\begin{proposition}[Greedy information-gain ordering]
\label{prop:order}
Given a matched set $C$ and a candidate coefficient $H_w$ not in $C$, the
conditional distribution of $H_w$ over the population of functions that
survive the observations in $C$ is
\[
\Pr[H_w=h \mid C] \;=\; \kc(C\cup\{H_w{=}h\})\,/\,\kc(C),
\]
which is computable from the
same recursion, and the expected one-step reduction in $\log_2\kc$ from
querying $H_w$ equals the Shannon entropy of this distribution.
Querying the maximum-entropy available coefficient is therefore the
greedy policy for expected one-step posterior gain.
\end{proposition}

The criterion in Proposition~\ref{prop:order} is not new and is not claimed
to be.  Selecting the next observation by expected information gain is
classical Bayesian experimental design~\cite{lindley1956, chaloner1995}.
Under a uniform prior over a finite hypothesis set with noiseless
deterministic tests it coincides with minimizing expected posterior
entropy~\cite{sebastiani2000}, with greedy version-space bisection, and with
generalized binary search~\cite{nowak2011}.  It should also be read as
\emph{greedy} rather than optimal.  The myopic rule is $O(\log
n)$-approximate with a matching hardness result~\cite{chakaravarthy2011}, it
is tightly $\Theta(\log n/\log \mathrm{OPT})$ under a uniform
prior~\cite{li2020}, and the optimal decision tree is NP-complete to
construct~\cite{hyafil1976}.  Adaptive submodularity~\cite{golovin2011}
supplies the conditions under which greedy is near-optimal, and is the
relevant theory here because $\kc$ depends on the observed \emph{values} and
not only on the set $C$, so the adaptivity is genuine rather than vacuous.
K\"opf and Basin apply the same greedy remaining-entropy rule, explicitly as
a heuristic rather than the optimum, to select the queries of an adaptive
side-channel attack~\cite{kopf2007}.

What Theorem~\ref{thm:dp} contributes is not the criterion but its exact
evaluability.  Bayesian experimental design ordinarily spends its effort
\emph{estimating} expected information gain, because the posterior that
defines it is intractable.  Here the posterior, and therefore the exact
information gain of every candidate coefficient, is available in closed form
before any coefficient is computed.

Two remarks on what Proposition~\ref{prop:order} means in practice.  First,
it accounts for the 2002 measurements.  Tables~V and~VI
of~\cite{thornton2002vlsi} report the number of coefficients examined before
a mismatch was found under random selection, and the averages range over
three decades across circuits.  Coefficients over small blocks have
low-entropy value distributions and are nearly worthless as evidence, and
random selection mixes them freely with the informative ones.  Second, the
policy is computable before committing.  The entropy of every available
coefficient is a by-product of the same tree recursion, so an implementation
can rank the not-yet-computed coefficients by expected information per unit
cost, where cost is the partial-BDD effort of obtaining them.  The 2002
method used whatever coefficients the diagram happened to expose.  The
procedure developed here decides which to compute before computing any of
them.

\subsection{Machine Verification}
\label{sec:verify}

Every formula above is verified by \texttt{j2\_verify\_v20.py}
(deterministic, seeded), included with this paper's artifact bundle.  (V1)
The conventions reproduce Table~IV of~\cite{thornton2002vlsi}, all sixteen
rows.  (V2) Theorem~\ref{thm:single} equals brute-force enumeration for every
level and value at $n=2,3$.  (V3) Theorem~\ref{thm:dp} equals brute force on
460 random constraint sets for $n=2,3,4$, exhaustively over all $65{,}536$
functions at $n=4$.  (V4) Theorems~\ref{thm:pair} and~\ref{thm:anc} agree
with the recursion and with brute force on random ancestor-closed sets and
300 $(H_0,H_1)$ pairs.  (V5) Corollary~\ref{cor:indep} holds on the paper's
example and on 100 random disjoint factorization identities.  (V6) Monotone
refinement holds over full random orderings with $\kc=1$ at full
conditioning.  All six groups pass.  No closed form appears in this paper
that a reader cannot re-derive by running the included script against
exhaustive enumeration.

The counting kernel used for the experiments of Section~\ref{sec:expt}
carries its own suite, which additionally checks two properties.  (V7)
Corollary~\ref{cor:indep} holds numerically, the ratio
$\kc/\kc_{\mathrm{ind}}$ being $1$ to nine digits on random disjoint sets at
$n=4,5$.  (V8) The basis-independence step of Proposition~\ref{prop:gap}
holds, in that $\sum_{B}S_B^{2}/|B|$ and $\sum_{i\in C}h_i^{2}/m_i$ agree,
over 240 random ancestor-closed sets at $n=3,\dots,6$.  The observed maximum
deviation in (V8) is exactly zero, as it must be, since that step is an
identity rather than an approximation.

\subsection{Experimental Results}
\label{sec:cecexpt}

The recursion, the closed forms, and the ordering rule are exercised here on
the same kind of data the 2002 method reported, with two purposes: to confirm
that the completed theory reproduces the earlier results exactly, and to show
what the exact version-space count reveals that the 2002 enumeration could not
reach.  All counts below are produced by the recursion of
Section~\ref{sec:cec} and, wherever the function count is small enough to
enumerate, checked against brute force; the driver is deterministic and is
included with the paper.

\paragraph{Reproducing the 2002 table}
Table~\ref{tab:cec_tableIV} regenerates Table~IV of the 2002
paper~\cite{thornton2002vlsi}: the modified-Haar spectrum $(H_0,H_1,H_2,H_3)$
of every two-variable function.  All sixteen rows agree with the original,
which fixes the sign and ordering conventions of Section~\ref{sec:haar} against
a published reference before any count is computed.

\paragraph{The closed forms against ground truth}
Table~\ref{tab:cec_verify} evaluates the two worked examples the 2002 paper
gave, the single coefficient $H_0=2$ at $n=2$ with $\kc=4$ and the
disjoint pair $H_2=0,\,H_3=-2$ at $n=2$ with $\kc=2$, together with a sweep of
the pair and single-coefficient closed forms at $n=4$.  The recursion, the
closed forms of Theorems~\ref{thm:single}--\ref{thm:anc}, and brute-force
enumeration agree in every case in which brute force is affordable, and the
recursion and the closed form agree in the $n=4$ cases whose $2^{16}$-function
brute force is omitted from the table.

\paragraph{The version space as evidence accumulates}
Table~\ref{tab:cec_shrink} is the measurement the 2002 method wanted but could
not compute in general: the exact size of the version space as coefficients are
conditioned one at a time, low-order-first, for three functions of $n=4$
variables.  Every column begins at $2^{16}=65536$ candidate functions and
descends to a single function once all sixteen coefficients are matched, and
every entry was verified against brute force.  The three columns descend at
very different rates, and the difference is the point.  The threshold function,
which is logic~$0$ on the first eight minterms and logic~$1$ on the last eight,
is pinned to a single function after only \emph{two} coefficients: $H_0=0$ and
$H_1=16$ together force the left half of the truth table to be all-zero and the
right half all-one, so no other function is consistent with them.  The parity
function and the random function, by contrast, retain thousands of candidates
after the same two coefficients and are resolved only as the small-support
coefficients near the leaves of the block tree are added.  A structured
function announces itself in its low-order spectrum; an unstructured one does
not.  This is exactly the behavior the greedy ordering rule of
Section~\ref{sec:seq} exploits, and it is invisible to a method that reports
only whether a fixed set of coefficients happened to match.

\subsection{Haar Matching Against Random Test Vectors}
\label{sec:cecspeedup}

Probabilistic equivalence checking competes, in practice, with the simplest
possible check: apply random input vectors to both circuits and declare them
different the moment an output disagrees.  The version-space view makes the
comparison precise.  A random test vector conditions one truth-table cell, so
it removes from the version space only the functions that disagree at that
single input; it detects a difference between two circuits only when it happens
to land on one of the inputs where they differ.  A low-order Haar coefficient
conditions a signed sum over a large dyadic block, so a difference anywhere in
that block moves the coefficient and is detected at once.

Table~\ref{tab:cec_speedup} quantifies this for the task of \emph{detecting}
inequivalence.  Two functions are made to differ in exactly $d$ minterms, and
we record how many observations each method needs to notice: for Haar matching,
the position of the first coefficient that differs when coefficients are
examined low-order-first; for random testing, the number of random input
vectors applied until one lands on a differing minterm.  The expectation of the
latter is $(N+1)/(d+1)$ for $N=2^n$ inputs, and the measured means match it.
The speedup is largest exactly where it matters most.  When the two circuits
differ in a single minterm, the subtlest possible discrepancy and the one a
random search is least likely to find, the global coefficient $H_0$ changes by
$\pm2$ and detects the difference on the first observation, against an expected
$(N+1)/2$ random vectors: a factor of several hundred at $n=10$ and several
thousand at $n=14$.  As $d$ grows the random search catches up, because a
larger set of differing minterms is easier to hit at random, and the
occasional cancellation of a difference in $H_0$ sends the Haar method one or
two coefficients deeper before it detects; the advantage narrows but remains
an order of magnitude or more across the range measured.

Two honest qualifications bound the comparison.  First, it is a comparison for
\emph{disproving} equivalence.  Proving equivalence outright requires driving
the version space to a single function, and both methods need all $N$
observations to do that, since the Haar coefficients and the truth-table cells
are each a basis of the same space; the value of the spectral method is the
asymmetry between the two directions, a single mismatched coefficient is a
proof of inequivalence while a single matched one is only weak evidence of
equivalence, and Table~\ref{tab:cec_shrink} is the record of how that weak
evidence compounds.  Second, the observation cost is counted here in
coefficients against test vectors, one apiece; a Haar coefficient is more
expensive to obtain from a decision diagram than a single simulation, so the
practical speedup is smaller than the observation-count ratio.  The
output-probability route of Section~\ref{sec:haar} keeps the per-coefficient
cost polynomial in the diagram size, which is what made the 2002 method
viable, but the honest claim is a large reduction in the \emph{number} of
observations needed to expose a difference, not an equal reduction in wall-clock
time.

\begin{table}[t]
\caption{Reproduction of Table~IV of the 2002 method~\cite{thornton2002vlsi}: the modified-Haar spectrum $(H_0,H_1,H_2,H_3)$ of every two-variable function, regenerated by the recursion of Section~\ref{sec:cec}.  The truth table lists $f(00)f(01)f(10)f(11)$.  All sixteen rows match the 2002 table.}
\label{tab:cec_tableIV}
\centering\small\setlength{\tabcolsep}{4pt}
\begin{tabular}{ccccc@{\quad}ccccc}
\toprule
$f$ & $H_0$ & $H_1$ & $H_2$ & $H_3$ & $f$ & $H_0$ & $H_1$ & $H_2$ & $H_3$\\
\midrule
\texttt{0000} & 4 & 0 & 0 & 0 & \texttt{1000} & 2 & -2 & -2 & 0\\
\texttt{0001} & 2 & 2 & 0 & 2 & \texttt{1001} & 0 & 0 & -2 & 2\\
\texttt{0010} & 2 & 2 & 0 & -2 & \texttt{1010} & 0 & 0 & -2 & -2\\
\texttt{0011} & 0 & 4 & 0 & 0 & \texttt{1011} & -2 & 2 & -2 & 0\\
\texttt{0100} & 2 & -2 & 2 & 0 & \texttt{1100} & 0 & -4 & 0 & 0\\
\texttt{0101} & 0 & 0 & 2 & 2 & \texttt{1101} & -2 & -2 & 0 & 2\\
\texttt{0110} & 0 & 0 & 2 & -2 & \texttt{1110} & -2 & -2 & 0 & -2\\
\texttt{0111} & -2 & 2 & 2 & 0 & \texttt{1111} & -4 & 0 & 0 & 0\\
\bottomrule
\end{tabular}
\end{table}
\begin{table}[t]
\caption{The counting recursion and the closed forms of Section~\ref{sec:cec} agree with brute-force enumeration and reproduce the two worked examples of the 2002 paper~\cite{thornton2002vlsi}.  A dash marks a case whose brute force ($2^{16}$ functions at $n{=}4$) is omitted; the recursion and closed form agree there and both agree with brute force wherever it is run.}
\label{tab:cec_verify}
\centering\footnotesize\setlength{\tabcolsep}{4pt}
\begin{tabular}{lrrrl}
\toprule
conditioned set & brute & recursion & closed form & source\\
\midrule
$n{=}2$, $H_0{=}2$ & 4 & 4 & 4 & Eq.~(24)\\
$n{=}2$, $H_2{=}0,\,H_3{=}{-}2$ & 2 & 2 & 2 & disjoint\\
$n{=}4$, $H_0{=}4,\,H_1{=}0$ & -- & 3136 & 3136 & Thm.~pair\\
$n{=}4$, $H_0{=}0,\,H_1{=}0$ & -- & 4900 & 4900 & Thm.~pair\\
$n{=}4$, $H_0{=}8,\,H_1{=}4$ & -- & 448 & 448 & Thm.~pair\\
$n{=}4$, $H(2,0){=}2$ & -- & 14336 & 14336 & Thm.~single\\
$n{=}4$, $H(3,0){=}0$ & -- & 24576 & 24576 & Thm.~single\\
\bottomrule
\end{tabular}
\end{table}
\begin{table}[t]
\caption{Version-space size $k(C)$ for three functions of $n{=}4$ variables as the $N{=}16$ modified-Haar coefficients are conditioned low-order-first.  Each column starts at $2^{16}=65536$ candidate functions before any coefficient is seen and descends to a single function when all sixteen are matched.  Every entry was checked against brute-force enumeration.}
\label{tab:cec_shrink}
\centering\small
\begin{tabular}{rrrr}
\toprule
coeffs conditioned $t$ & parity & threshold & random\\
\midrule
$0$ & $65536$ & $65536$ & $65536$\\
$1$ & $12870$ & $12870$ & $11440$\\
$2$ & $4900$ & $1$ & $3920$\\
$3$ & $2520$ & $1$ & $2016$\\
$4$ & $1296$ & $1$ & $864$\\
$5$ & $864$ & $1$ & $576$\\
$6$ & $576$ & $1$ & $384$\\
$7$ & $384$ & $1$ & $256$\\
$8$ & $256$ & $1$ & $128$\\
$9$ & $128$ & $1$ & $64$\\
$10$ & $64$ & $1$ & $32$\\
$11$ & $32$ & $1$ & $16$\\
$12$ & $16$ & $1$ & $8$\\
$13$ & $8$ & $1$ & $4$\\
$14$ & $4$ & $1$ & $2$\\
$15$ & $2$ & $1$ & $2$\\
$16$ & $1$ & $1$ & $1$\\
\bottomrule
\end{tabular}
\end{table}
\begin{table}[t]
\caption{Observations required to \emph{detect} that two functions differ, by Haar-coefficient matching against random test vectors, versus the number $d$ of differing minterms ($400$ trials per row).  \emph{Haar}: mean position of the first differing coefficient, low-order-first.  \emph{Random}: mean number of random vectors until one hits a differing minterm, with exact expectation $(N{+}1)/(d{+}1)$.}
\label{tab:cec_speedup}
\centering\small
\begin{tabular}{rrrrrr}
\toprule
$n$ & $d$ & Haar & random & $(N{+}1)/(d{+}1)$ & speedup\\
\midrule
$10$ & $1$ & $1.00$ & $539.4$ & $512.5$ & $512\times$\\
$10$ & $2$ & $5.57$ & $334.3$ & $341.7$ & $61\times$\\
$10$ & $4$ & $1.99$ & $192.8$ & $205.0$ & $103\times$\\
$10$ & $8$ & $1.39$ & $115.5$ & $113.9$ & $82\times$\\
$10$ & $16$ & $1.22$ & $59.1$ & $60.3$ & $49\times$\\
\midrule
$14$ & $1$ & $1.00$ & $8432.2$ & $8192.5$ & $8192\times$\\
$14$ & $2$ & $6.45$ & $5508.6$ & $5461.7$ & $846\times$\\
$14$ & $4$ & $1.84$ & $3334.9$ & $3277.0$ & $1781\times$\\
$14$ & $16$ & $1.23$ & $1010.9$ & $963.8$ & $785\times$\\
$14$ & $64$ & $1.12$ & $248.3$ & $252.1$ & $225\times$\\
\bottomrule
\end{tabular}
\end{table}
\FloatBarrier

\section{Application II: Key Counting for Locked Netlists}
\label{sec:lockapp}
\label{sec:lock}

The same version-space machinery moves into hardware security with almost
nothing changed but the names of its parts.  What was a hidden Boolean function
is now a secret key, one of $2^k$ candidates; what were spectral coefficients
are now input--output pairs read from an unlocked oracle chip.  Each such pair
keeps the keys that reproduce the observed response and discards the rest, so
the keys still standing after a query set form a version space, and its
logarithm, the surviving entropy, is on the same axis as the key length the
lock advertises.  The counting engine is again the one of
Section~\ref{sec:cec}, run this time over the gate-level factor graph of the
netlist in place of the Haar block tree, and the reparameterization that made
the spectral count local reappears in a new guise: a combinational netlist
fixes its own internal signals, so the count is plain model counting and not
projected model counting.  The rest of the section puts the machinery on a
locked netlist, establishes that the count is not a projected count, and names
the graph width that sets its cost.

The three subsections that follow sketch logic locking, the oracle-guided
attack, the threat model they share, and what the security community reports
today.  None of this is new, and a reader at home in hardware security can
proceed directly to Section~\ref{sec:constraint}.

\subsection{Locked Netlists and the Oracle-Guided Attack}

A locked netlist $C(X,K)$ has primary inputs $X$ and key inputs $K$, with a
secret key $K_c$ for which $C(X,K_c)=C_o(X)$ for every $X$, where $C_o$ is the
original design.  The constructions differ mainly in what the key
controls.  Random logic locking inserts XOR and XNOR key gates at randomly
chosen nets, so that a wrong key bit inverts the signal on its
net~\cite{roy2008epic}.  Strong logic locking places those gates so that they
interfere with one another as much as possible~\cite{yasin2016sll}, which
defeats sensitization, the attack that isolates one key bit at a time by
propagating it to an output while blocking the others~\cite{rajendran2012}.
Point-function schemes such as SARLock~\cite{yasin2016sarlock} and
Anti-SAT~\cite{xie2019antisat} add a comparator that corrupts the output on a
single input pattern per wrong key, and stripped-functionality logic locking
removes a set of input cubes and restores them with a key-controlled
unit~\cite{yasin2017sfll}.  Older work calls all of this logic encryption;
we follow Engels, Hoffmann and Paar in preferring logic locking, since nothing
is encrypted~\cite{engels2022}.

The threat model assumes the attacker holds two artifacts, a copy of the
locked netlist obtained from the foundry and an unlocked chip that serves as
an oracle.  The SAT attack of Subramanyan, Ray and Malik turns this into a
procedure~\cite{subramanyan2015}: it builds a miter from two copies of the
locked netlist sharing primary inputs but carrying independent keys, asks a
solver for an input on which the copies disagree, applies that distinguishing
input to the oracle, and adds the constraint that any surviving key must
reproduce the response, ending when no distinguishing input remains.  The
attack is usually described in terms of keys but really operates on
equivalence classes, as Sirone and Subramanyan say directly when they describe
it as working ``by using a Boolean SATisfiability solver to iteratively find
inputs that distinguish between equivalence classes of
keys''~\cite{sirone2020}.

\subsection{The Threat Model and Its Assumptions}
\label{sec:threat}

Because the number this paper reports is defined against a specific attacker,
it is worth setting down exactly what that attacker holds, what it may do, and
what the netlist is assumed to be, so that the surviving entropy can be read
without ambiguity about the conditions it was measured under.

The attacker holds two artifacts and nothing else.  The first is the locked
netlist itself, at gate level, which in the standard supply-chain story leaks
from an untrusted foundry that must be given the design in order to fabricate
it.  The second is an activated chip, bought on the open market, whose key
register has been programmed with the correct value and sealed.  That chip is
an oracle: the attacker drives its primary inputs with a pattern and reads the
correct primary outputs, and can repeat this for as many patterns as it cares
to apply.  What the attacker cannot do is read the key out of the chip, probe
an internal net, or observe anything beyond the input--output behaviour; the
oracle is a black box whose only leak is the function it computes.  This is the
model under which every scheme cited above advertises its security, and it is
the model measured here.

A query in this paper is one such input--output pair.  Two regimes are
distinguished throughout and they cost the attacker very different things.  A
\emph{random} query is an input pattern drawn without reference to the netlist,
which is free: it is one evaluation of a chip the attacker already owns.  A
\emph{chosen} query is a pattern produced by reasoning about the locked netlist
so as to separate keys the earlier queries left standing, which is what a
satisfiability-based attack does at every step, and it is not free, because the
reasoning is a solver call whose cost tracks the very property that makes the
lock hard.  The surviving entropy is a function of the query set, so it is
reported against both regimes, and the gap between them is itself measured in
Section~\ref{sec:adversarial}.

The method assumes the netlist is combinational and acyclic.  Constant
propagation and the topological argument of Lemma~\ref{lem:notproj} both rely
on it, and it is why cyclic locking is declared out of scope rather than run
and reported: a netlist with combinational feedback has no topological order in
which each internal net is determined by earlier ones, and the count is not
defined the same way.  It assumes further that the key ports can be identified
in the netlist, which is a reader's problem rather than a counter's and is
treated as such in Section~\ref{sec:validation}, and that the gate-level
structure is available, which the threat model already grants because the
foundry copy is structural.

The relation to the SAT attack is a matter of what is being produced.  The SAT
attack runs to a key, or to a timeout, and reports how long it took or how many
iterations it needed; this method runs to a count, and reports how many keys an
attacker with the same query set cannot yet tell apart.  The two meet at a
single condition.  The SAT attack halts exactly when no input pattern
distinguishes two surviving keys, and that is precisely the condition under
which the count this paper reports has stopped being an upper bound and become
exact.  The attack's termination test and the paper's tightness certificate are
the same satisfiability question, put to different ends, and
Section~\ref{sec:floor} makes that identity explicit.

One caveat belongs with the model rather than with the method.  Some schemes do
not use the key bits directly but route them through a fixed preprocessing
stage before they reach the logic, so that the object the attacker recovers is
a preprocessed key rather than the raw one.  The count reported here is of raw
key values consistent with the oracle.  When the preprocessing is injective the
two coincide, and when it is not the raw count is an upper bound on the number
of distinguishable effective keys, in the same conservative direction as every
other bound in this paper.  We flag this so that a surviving-entropy figure is
not read as a claim about a key-derivation stage the method does not inspect.

\subsection{Past Reported Results and Omissions}

Four quantities carry most security claims.  \emph{Key length}, or $2^k$ as a
key space, is the original claim and is still widely used.  \emph{Iteration
count}, the number of distinguishing patterns the SAT attack needs, became the
operative number after 2015, and point-function schemes are advertised by it:
their design goal is to force an iteration count exponential in the key
length, achieved by arranging that each pattern eliminates only one key.
\emph{Output corruptibility} is the fraction of input and wrong-key
combinations on which the locked circuit differs from the original, commonly
targeted at a Hamming distance of one half.  \emph{Solver runtime} is the time
an attack takes before it completes or times out, and Sweeney, Heule and
Pileggi name the problem with it: ``It is common to only report attack run
times, however, the keys that can be obtained at timeout may be very close to
a correct solution''~\cite{sweeney2020}.  None of the four reports how many keys the
attacker still cannot rule out.

Two recent surveys confirm the gap: neither the
review in~\cite{azar2019} nor the more recent survey of the
field~\cite{kamali2022advances} reports a key count or an entropy for any
scheme, and the latter observes that ``generic and comprehensive security
metrics are completely missing.''  Work that formalizes the problem does so
definitionally rather than numerically, with simulation-based and
indistinguishability notions satisfied by universal
circuits~\cite{masserova2022logic,beerel2022}, settling what security would
mean without saying how far a given netlist is from it.  The one composite
metric we know of combines corruption, structural resilience and a residual
key-space term measured by running attacks under a time
budget~\cite{sisejkovic2018}; it is attack-relative, where what follows is an
exact count conditioned on a given query set.

Counting has been brought to bear on locking once before, differently.
Sweeney, Garg and Pileggi estimate the fraction of the key space whose
corruption clears a threshold, using the approximate counter
ApproxMC~\cite{sweeney2024}, and note that ``exact solutions can only be
obtained for a limited key and input width.''  That quantity is a property of
the lock alone.  It is not conditioned on anything an attacker has observed,
so it does not move as queries accumulate, and it is not the number this paper
computes.

\subsection{Relation to Quantitative Information Flow}
\label{sec:qif}

The measure reported here is a quantitative information flow measure, and the
machinery for computing such measures by counting is established.  Backes,
K\"opf and Rybalchenko recover the partition a program induces on its secret
inputs and count each block~\cite{backes2009}; Klebanov, Manthey and Muise do
the same by reducing the program to a Boolean formula and applying a model
counter~\cite{klebanov2013}; and K\"opf and Basin define an adaptive
attacker's residual uncertainty as a function of the query count and plot it
against that number~\cite{kopf2007}.  What this paper does is that
computation, for a locked netlist under an oracle-guided attack.

That combination has not been made.  Quantitative information flow has reached
hardware, but measures a different thing: QIF-Verilog and its successors
quantify how much a design leaks about the data it processes, such as a
cryptographic key reaching an observable output~\cite{guo2019}.  The secrecy
of a locking key under oracle access is not what those tools are pointed at,
and the one study bringing the two communities together evaluates locking by
path sensitization instead, closing with the observation that ``quantitative
information flow analysis methods may identify additional
vulnerabilities''~\cite{reimann2025}.

Two further effects make that omission worse.  The first is that keys can be
functionally equivalent, which reduces the key space below $2^k$ by an amount
that is not always small.  The effect is documented: Sweeney et al.\ report
that in a banyan locking network ``multiple keys produce the same permutations
of the inputs on the outputs''~\cite{sweeney2020}, Hu et al.\ give a
three-bit SARLock instance in which three incorrect keys each unlock part of
the input space~\cite{hu2024}, and Zhong and Guin state the general position
that ``if multiple keys remain in the search space, it must be true that
multiple solutions are valid since they all give the same output
response''~\cite{zhong2023}.  What is missing in each case is a number.

The second effect is the structural premise this paper builds on.  Zhong and
Guin, studying SAT attack complexity, observe that ``a response 0 at the OR
gate effectively splits the logic cone into two subcones, where keys inside
the subcone are dependent, but independent from the other
subcone's''~\cite{zhong2023}.  That is a conditional independence statement
about key bits, found empirically and left unexploited, and sum--product
message passing is the standard way to exploit statements of that form;
Juretus and Savidis point the same way when they organize their analysis of
locking resilience around logic cones and fanout-free
regions~\cite{juretus2020}.

\subsection{The Constraint System}
\label{sec:constraint}

Take as given a locked netlist $C(X,K)$ and a query set
$\{(x_i,y_i)\}_{i=1}^{t}$ collected from the oracle.  The quantity to compute is
$|V_t|$, the size of the version space
\[
V_t=\bigl\{\,K\in\{0,1\}^{k} \;:\; C(x_i,K)=y_i,\ i=1,\dots,t\,\bigr\}.
\]

For a single query we set the primary inputs to $x_i$ and let constants
propagate through the netlist.  A gate all of whose inputs are now fixed becomes
a constant in turn, and a gate left with one symbolic input becomes a buffer or
an inverter, so the sweep prunes the netlist down to a residual circuit over
the key bits alone.  That residual is usually a small part of the whole design,
for the plain reason that most of a circuit computes something the key does not
touch.

Every surviving gate $g$, with output net $w$ and symbolic inputs
$u_1,\dots,u_r$, then contributes one indicator factor
\begin{equation}
\psi_g(w,u_1,\dots,u_r)\;=\;\bigl[\,w=\mathrm{op}_g(u_1,\dots,u_r)\,\bigr],
\label{eq:gatefactor}
\end{equation}
while the observed response fixes each primary output net to its value in
$y_i$.  Carrying this out for all $t$ queries, keeping the internal nets private
to each query but sharing the key bits across them, yields a single factor set
whose satisfying assignments are exactly the keys of $V_t$, each paired with the
internal signals that key forces.

\subsection{Key Counting is Not Projected Counting}
\label{sec:notprojected}

Read naively, the problem looks far more expensive than it is, and the gap
between the naive reading and the true one is large enough to be worth a lemma,
even though the principle the lemma rests on is standard.

\begin{lemma}
\label{lem:notproj}
Let $\Phi$ be the factor set above, over the key variables $K$ and the
residual internal variables $W$.  Then the number of satisfying assignments
of $\Phi$ over $K\cup W$ equals $|V_t|$.
\end{lemma}

\begin{IEEEproof}
The netlist is combinational and acyclic.  Fix a key $K$ and a query index
$i$.  Every residual internal net of query $i$ is the output of a gate whose
inputs are either constants determined by $x_i$ or nets earlier in
topological order, so by induction each internal net takes exactly one value.
The factors in~\eqref{eq:gatefactor} are satisfied by that value and by no
other.  Hence each $K\in V_t$ extends to exactly one satisfying assignment,
and each $K\notin V_t$ extends to none.
\end{IEEEproof}

Counting keys therefore requires no projection.  The principle is not new and
we do not claim it.  Aziz et al.\ state it in passing when introducing
projected counting, observing that counting the models of a formula gives the
right answer whenever the added variables are functionally defined by the
original ones~\cite{aziz2015}, and Lagniez, Lonca and Marquis prove it: if the
variables to be forgotten are definable from those retained, forgetting them
leaves the count unchanged~\cite{lagniez2020}.  The same phenomenon carries
the name \emph{independent support} in the sampling and counting
literature~\cite{ivrii2016}.  Lemma~\ref{lem:notproj} is the instance
in which the definability is not discovered by a solver but given by the
netlist: a combinational circuit determines its own internal signals, by
construction, for every key.

The consequence is a complexity statement that neither literature has drawn on
its own.  Existentially quantifying the internal signals away turns the problem
into projected model counting, whose algorithms are doubly exponential in
treewidth, a bound that is tight under the exponential time
hypothesis~\cite{fichte2022}.  Leaving the internal signals in place keeps it
plain model counting, singly exponential in width.  Key counting always lands
on the plain side of that line and never on the projected side, and at the
widths Section~\ref{sec:expt} measures that placement is the whole difference
between a count one can actually compute and one there would be no point
attempting.

\subsection{Two Widths}
\label{sec:twowidths}

What sets the cost is the induced width of the elimination order, and there are
two graphs the order can run over.  The \emph{key-moral width} takes the graph
on the key variables alone, reached by eliminating every internal variable
first, so that two key bits become neighbours the moment some internal
structure ties them together; it is the graph a practitioner who reasons about
relations among key bits directly would picture.  The \emph{factor width} keeps
the internal variables in the graph and takes the induced width of the
gate-level factor graph itself.

The two are not close.  Take a point-function lock, in which all $k$ key bits
are compared against a pattern by a comparator whose outputs feed an AND-tree
ending in one primary output.  Eliminating the internal nets of that tree makes
every key bit adjacent to every other, so the key-moral graph becomes a clique
on $k$ vertices and its width is $k-1$.  But the AND-tree is a tree: retain its
internal nets and every factor has scope three, and the width of the factor
graph does not grow with $k$ at all.

Interference-maximizing placement behaves the same way.  Strong logic locking
positions its key gates to interfere as much as possible, precisely so that no
key bit can be isolated~\cite{yasin2016sll}, which is a claim about pairwise
sensitization; it makes the key-moral graph dense and says nothing whatever
about the width of the gate-level factor graph.  Section~\ref{sec:expt} shows
the two widths pulling far apart on exactly these constructions.

\subsection{A Second Representation}
\label{sec:engineb}

The same integer can be reached a second way, and the second way has the
opposite cost trend, which is what makes it worth carrying.  With the primary
inputs of a query fixed, simulate the netlist symbolically until every net is a
function of the key bits alone; conjoin, for each observed output, the agreement
between that net and the value seen, and then conjoin those constraints across
the queries.  What results is a Boolean function whose models are exactly $V_t$,
and its count is a single traversal of the decision diagram that holds it.

Read against $t$, the two costs run in opposite directions.  Every query adds
one more copy of the residual circuit to the factor graph, so the elimination
width climbs with $t$; every query further constrains the decision diagram,
which represents the version space itself and so contracts with $t$.  Which
representation is the cheaper is therefore a question about the object being
measured and not about either implementation, and Section~\ref{sec:engines}
settles it on real benchmarks.

The count on the diagram has to be exact.  A minterm count returned as a
double-precision float is already wrong above $2^{53}$, comfortably inside the
range a $64$-bit key space occupies, so throughout, the C engine uses
\texttt{Cudd\char`_Apa\-Count\-Minterm} in place of
\texttt{Cudd\char`_Count\-Minterm}.

\subsection{Failure Analyses of the Methods}
\label{sec:lockfail}

Two constructions ought to be hard for these methods, and naming them in
advance is the honest way to report where a method breaks.

The first is routing-network locking.  Full-Lock~\cite{kamali2019} and
InterLock~\cite{kamali2020} hand the key control of a multistage permutation
network, and banyan and butterfly topologies are the textbook worst cases for
graph width; InterLock in particular is built so that its routing constraints
resist being pulled apart for separate analysis.  If the method fails anywhere,
this is the family in which one would expect it to fail.

The second is width growth with the query count.  A query contributes a fresh
copy of the residual circuit that shares the key bits, and where the residual
is small those copies scarcely interact.  A point-function lock is the opposite
case: every query reaches every key bit, so the copies stack and the factor
width climbs with $t$, and because the scheme is arranged so that a query
eliminates just one key, a great many queries are needed and the two pressures
reinforce each other.

This second prediction concerns Engine~A alone and does not carry to Engine~B.
After $t$ queries the version space of a point-function lock is the full key
space minus at most $t$ points, and a diagram that excludes $t$ points grows
only linearly in $t$, never exponentially.  Engine~B's representation is thus
smallest just where Engine~A's width is largest.  If Section~\ref{sec:engineb}
has the argument right, Engine~A should halt on a point-function lock within a
dozen or so queries while Engine~B sails past it untouched, and
Section~\ref{sec:expt} reports both halves.

\section{Experimental Evaluation}
\label{sec:expt}

The evaluation runs two campaigns.  E.1 examines the structure of the counting
problem on netlists locked for this paper, where we know the scheme
implementation exactly; E.2 measures the published locked benchmarks, where we
do not.  Where the two disagree, E.2 is the one that governs.

Every count is produced by the engines of Section~\ref{sec:impl}, which share
no representation: Engine~A eliminating variables over the residual factor graph
and paying in induced width, Engine~B carrying a decision diagram over the key
variables and paying in diagram size, in both a self-contained Python form and a
C form on CUDD~\cite{cudd}.  Below $k=20$ a brute-force sweep of all $2^k$ keys
supplies ground truth as well.  Because the implementations are independent,
their agreement is evidence rather than a single computation restated, so every
instance more than one of them can reach is checked by all that can.

\subsection*{E.1 Experiments: Locked Netlist Examples}

The netlists used for the E.1 experiments are generated ripple-carry adders, array
multipliers of the c6288 family, and equality comparator trees.  The locking
schemes are our own reimplementations in the style of random insertion,
interference-maximizing placement and the point-function family, rather than
the original authors' code, and every width and count reported is exact for
the instance given.

\subsection{Validating the Counter}
\label{sec:e10}

The counter is checked against exhaustive enumeration over all $2^{k}$ keys,
on adders, comparators and multipliers, under all three locking schemes, at
key lengths three to five and query counts one to three: eighty-one cases with
zero mismatches.  The check is inexpensive, and it is what lets the rest of
the numbers be read as measurements.

\subsection{Measuring the Two Widths}
\label{sec:e11}

\paragraph{Setup}  Three benchmark families and three locking schemes, at key
lengths of $8$, $16$, $24$ and $32$ bits, with four oracle queries per
instance.  Both of the widths defined in Section~\ref{sec:twowidths} are
computed from a greedy minimum-fill elimination order.

\paragraph{Result}  The width over key variables alone tracks the key length.
For interference-maximizing and point-function locking it is $k-1$ almost
everywhere: $7$ at eight key bits, $15$ at sixteen, $23$ at twenty-four and
$31$ at thirty-two.  That is a clique, and it is a clique for the reason the
schemes intend.  The factor width behaves quite differently.  For
point-function locking it ranges from $6$ to $12$ across the whole sweep,
while the key-moral width over the same instances ran from $7$ to $31$; at $32$
key bits, where the key-moral width is $31$, the factor width is $10$ on
adders and on comparators.  For random locking it ranges from $2$ to $14$.
Fig.~\ref{fig:e1width} plots the two against each other.

\paragraph{The hardest case}  Interference-maximizing locking on a $4\times4$
array multiplier reaches factor width $18$ at thirty-two key bits, against a
key-moral width of $31$, and the same factor width $18$ at twenty-four against
$23$.  These are the configurations in the sweep where the two widths come
closest, and we report them because they bound the claim: array multipliers
have dense reconvergent structure of their own, and combining that with
placement chosen to maximize interference removes most of the margin.

\begin{figure}[t]
\centering
\includegraphics[width=0.86\columnwidth]{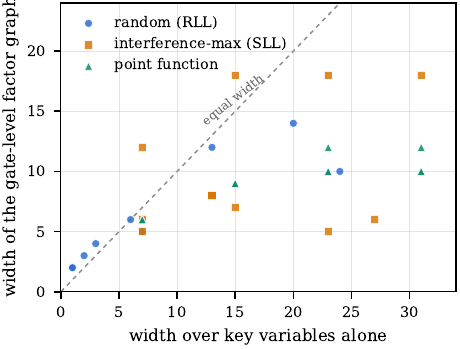}
\caption{Width over key variables alone, against width of the gate-level
factor graph, for every instance in the campaign E.1 width sweep.  Points below the
diagonal are instances where retaining the internal signals is cheaper than
eliminating them.  Point-function locking sits far below the diagonal and does
not move right to left as the key lengthens.  At $32$ key bits the key-moral
width is $31$ and the factor width is $10$ to $12$.  The array-multiplier
points under interference-maximizing locking carry the largest factor widths
in the sweep, reaching $18$, and are the hardest case found.}
\label{fig:e1width}
\end{figure}

\subsection{Version-Space Trajectories}
\label{sec:e12}

\paragraph{Setup}  An eight-bit ripple-carry adder locked with twenty-four key
bits under each of the three schemes, queries drawn uniformly at random and
applied in sequence, with $|V_t|$ computed exactly after each one.

\paragraph{Three behaviors}  The three schemes behave qualitatively
differently, and the trajectories say so plainly.  Under point-function
locking the version space falls from $16{,}777{,}215$ to $16{,}777{,}214$ to
$16{,}777{,}213$ and so on down, one key per query, which is the defining
property of the family and the reason such schemes force an attack iteration
count exponential in the key length.  Under interference-maximizing locking
the count collapses through
$131{,}072$, $12{,}288$, $48$, $6$, $3$, $3$, $3$, $1$, so eight random
queries identify the key uniquely.  The scheme is built to resist
sensitization and it does, but it offers very little to an oracle-guided
attacker who simply asks.  Under random locking the count falls and then
stops, at $98{,}304$, $98{,}304$, $49{,}152$, $16{,}384$, $16{,}384$, $8{,}192$
and then $8{,}192$ for every further query out to twelve.  Those $8{,}192$
survivors are not separable by any input pattern because they are functionally
equivalent, so thirteen of the twenty-four key bits carry no secret at all.
This is the equivalent-key effect that~\cite{sweeney2020}
and~\cite{hu2024} describe, here given as a number.

\paragraph{Engine halting points}  The failure mode predicted in
Section~\ref{sec:lockfail} appears exactly as described.  Under point-function
locking the factor width grows with query count, going $2$, $3$, $5$, $10$,
$15$, $19$, $21$, $23$, and passing the time budget at the eighth query,
because each query adds a copy of the comparator tree and every copy touches
every key bit.  Interference-maximizing locking follows at the twelfth query
and random locking survives to the ninety-eighth.

The other half of the prediction is confirmed as well, and it governs what we
can claim.  Run on the same instances and queries, Engine~B reaches three
hundred queries on all three schemes.  On the
point-function lock its accumulated diagram is $163$ nodes at ten queries,
$1{,}026$ at a hundred and $1{,}875$ at three hundred, for five seconds of
work in total, while the version space is still of size $16{,}776{,}916$; on
the interference-maximizing lock the diagram falls to $24$ nodes and stays
there once the key is identified.  Wherever both engines answer in E.1 they
agree.  So where the one-key-per-query rate is observed for seven queries
under Engine~A, Engine~B measures it over three hundred with $|V_t| = 2^{24} - t$
holding exactly throughout.  The two engines are complementary rather than
redundant: Engine~A is inexpensive when the per-query residual circuit is
small and expensive otherwise, while Engine~B is indifferent to that and
expensive instead when the version space itself is large, which is the failure
mode Section~\ref{sec:engines} exhibits on a published benchmark.

\subsection*{E.2 Experiments: Published Benchmarks}

The E.2 experiments take the Trust-Hub obfuscation release~\cite{trusthub} as
their input.

\subsection{Validating the Readers and the Engines}
\label{sec:validation}

On every invocation the harness runs six checks: the reader round trip on all
ten ISCAS-85 circuits in two independent formats; key recovery, meaning that a
locked netlist under its published key reproduces the reference on a thousand
random vectors; Engine~A against brute force; the wide-gate decomposition
against brute force on four hundred cases; Engine~A against Engine~B against
brute force on two hundred and forty cases; and the tightness certificate of
Section~\ref{sec:floor} against exhaustive enumeration of the version space
followed by an exhaustive search over input patterns.  The two readers are
also cross-checked functionally: the ISCAS \texttt{.isc} reader and the
structural Verilog reader parse different files describing the same circuit
and agree on random vectors for c432, c880 and c6288, the last with 6,256
gates against 2,353 and all thirty-two outputs compared.

The structural Verilog reader is checked twice more against tools written by
others.  It is compared against Icarus Verilog~\cite{icarus} on the release
files themselves: each instance is read, random primary-input and key vectors
are drawn, and the two evaluations are compared output by output.  Over the
whole release, $233$ instances agree and none disagrees.  Of the remainder,
$43$ are rejected for combinational feedback, which is the correct outcome and
matches the count of cyclic instances the release advertises, and $19$ are the
malformed files described in Section~\ref{sec:screen}.  Separately, the number
of key ports the reader identifies is compared against the key size each
instance's own documentation declares, and agrees on all $233$.  Nothing in
the campaign rests on a reader that has not been checked against an independent
implementation of the same language.

\subsection{Reachable Instances}
\label{sec:screen}

An instance is screened for reachability before it is used, and the screen runs
twice, because what the two engines can reach differs.  The Engine~A screen
constant-propagates each query, assembles the residual factor graph, and asks
whether a greedy order holds the induced width at or under a threshold; the
Engine~B screen builds the accumulated diagram for the same queries and asks
whether it stays within a node budget.  Neither screen ever computes a count.

\begin{table}[t]
\caption{Reachability of the c432 and c1355 families at three
queries.  $n$ is the number screened, and the A and B columns count
how many of them each engine reaches: Engine A when a greedy order
holds the residual width at or below $25$, Engine B when its diagram
stays inside its node budget.  The two disagree, and in opposite
directions as the key lengthens.}
\label{tab:screen}
\centering
\small
\setlength{\tabcolsep}{4pt}
\begin{tabular}{@{}lrrrrl@{}}
\toprule
family & bits & $n$ & A & B & A widths\\
\midrule
RN, SL, CS & 32 & 6 & 5 & 6 & 8--24\\
 & 64 & 6 & 2 & 4 & 14--17\\
 & 128 & 6 & 1 & 1 & 17--17\\
 & 256 & 3 & 0 & 0 & --\\
AntiSAT & 140--379 & 21 & 6 & 0 & 16--25\\
BDD-based & 28--32 & 3 & 3 & 3 & 6--17\\
\bottomrule
\end{tabular}
\end{table}

Table~\ref{tab:screen} reports the c432 and c1355 families of the published
suite at three queries, with the threshold at $25$ and the node budget at
eight million.  At $k=32$ Engine B reaches every instance, Engine A reaches
five of the six, and the widths that stay under the threshold run from $8$ to
$24$.  Above that the two engines part company, and in both directions:
across the whole screen there are eight instances Engine A reaches and
Engine B does not, and five that Engine B reaches and Engine A does not.
Nine are reached by both and twenty-three by neither.

The direction of the split is the informative part.  Engine~A's width barely
moves with key length, sitting at $16$ to $25$ even on the AntiSAT variants
carrying $140$ to $379$ key bits, as Section~\ref{sec:twowidths} predicts;
what defeats it is the query count.  Engine~B fails on those same instances
because their version spaces remain enormous after three queries.  A screen at
$t=3$ is therefore close to the worst case for Engine~B, whose cost falls as
queries accumulate, and close to the best case for Engine~A, whose cost rises,
so it understates what Engine~B will reach in a full campaign.  That is the
right direction for a screen: it declines instances it would in fact manage
rather than promising ones it would not.  Four of the six $k=64$ instances
screened here appear in Table~\ref{tab:phase6}, and Engine~B reaches all four
at $t=3$, three of them where Engine~A does not.

Two groups are not used, for different reasons.  Cyclic locking, 43 instances
across the release and nine within the two families screened here, introduces
combinational feedback and is outside the scope of a method that assumes an
acyclic netlist.  The synthesized variants are mapped to a standard cell
library and would need that library before they could be read; every instance
of the release also ships an unsynthesized netlist, and that is the file used
here.  The BDD-based instances do appear in Table~\ref{tab:screen}, and both
engines reach all three of those falling within the two families screened.
They were excluded from earlier versions of this work on the ground that they
exposed no key ports, which turned out to be a property of the reader rather
than of the benchmarks: those instances name their key ports
\texttt{keyinput\textit{n}\_\textit{GateName}} rather than
\texttt{keyinput\textit{n}}, a convention their own documentation states.
Section~\ref{sec:bdd} reports the whole family.

\subsection{Failure Modes for Each Engine}
\label{sec:engines}

Before any data existed, the development plan for this work wrote down a
prediction: that some benchmarks would let Engine~A through while Engine~B ran
past its node budget, and that c6288 was the natural place to look for them.
The regime turned out to be real and the circuit turned out to be wrong, and we
report both findings.

\begin{table}[t]
\caption{Engine A against Engine B on published benchmarks, on the
same queries.  The two engines agree exactly wherever both complete.
$t_A$ and $t_B$ are the last query each engine finished; a value
equal to the query budget means the engine never failed.  The width
and diagram columns are taken at $t_A$ and at $t_B$ respectively,
and the diagram is the accumulated constraint, which is the version
space itself.}
\label{tab:engines}
\centering
\small
\setlength{\tabcolsep}{4pt}
\begin{tabular}{@{}lrrrrrr@{}}
\toprule
benchmark & $t_A$ & width & A s & $t_B$ & diagram & B s\\
\midrule
\bench{c1355-RN320} & 4 & 23 & 80.8 & 6 & 30 & 0.0\\
\bench{c1355-SL320} & 6 & 18 & 3.6 & 6 & 165 & 0.0\\
\bench{c432-RN640} & 3 & 26 & 568.7 & 2 & 3,661,940 & 9.7\\
\bench{c432-SL320} & 4 & 24 & 84.8 & 6 & 64 & 0.0\\
\bench{c880-RN320} & 6 & 14 & 0.4 & 6 & 35 & 0.0\\
\bottomrule
\end{tabular}
\end{table}

On four of the five benchmarks in Table~\ref{tab:engines}, Engine B is ahead
at every query count and the margin widens with $t$, and on two of them
Engine A stops finishing before the sixth query.  The fifth column is the
reason: it is the size of the accumulated constraint, which is the version
space, and on those four it stays between $30$ and $425$ nodes at well under a
tenth of a second, while Engine A's running time rises by three orders of
magnitude and on two of them exhausts its four-minute budget.

The fifth benchmark, \bench{c432-RN640}, runs the other way.  There Engine A
answers the first query in a hundredth of a second at width $5$ and the
second in two seconds at width $17$, while Engine B's accumulated diagram is
already at $152{,}097$ nodes and then $3{,}661{,}940$, taking ten seconds.
The predicted regime therefore exists and is not hypothetical.  What the
prediction got wrong is only the circuit.  At the third query Engine B
exceeds its node budget while Engine A finishes at width $26$ in nine and a
half minutes, so on this instance Engine A does extend the reachable range
and the two engines fail at opposite ends exactly as predicted.

The qualification in Section~\ref{sec:engineb} is the whole of
\bench{c432-RN640}: Engine~B's diagram shrinks with $t$ only once the version
space is small enough to represent, and while $|V_t|$ is still above $2^{52}$
that trend has not begun.  Where it has begun, the crossover is not at some
large $t$ but at the first query.

The practical consequence revises the emphasis of
Section~\ref{sec:twowidths} without contradicting it.  Retaining the internal
signals is what makes the problem tractable at all, and that argument stands;
but the object to represent is the version space rather than the constraint
system, and the cost parameter that matters in practice is the diagram size of
$V_t$ rather than any width of the query system.  That $V_t$ should have a
small representation is not a coincidence, since it is the object the method
exists to measure.

\subsection{The Value of a Query}
\label{sec:queryworth}

Queries do not all carry the same information, and here too the spread across
instances dwarfs the spread across schemes.  For each of twenty random input
patterns we record the fraction of a two thousand key sample the query removes,
and take the median.  On seven of the eight $32$-bit instances measured that
median falls between $0.868$ and $1.000$, and not one pattern removes nothing.
The eighth is the exception: on \bench{c1355-CS320} the median is $0.189$, and
\emph{ten of twenty random inputs eliminate no key at all}.

That this is the instance and not the scheme is settled by the other cone-size
instance, \bench{c432-CS320}, whose median is $0.934$ with no worthless pattern
among the twenty.  Instances of the first kind therefore exist in the published
suite, and an iteration count, which registers a query either way, cannot tell a
query worth a bit from a query worth nothing.

\subsection{Advertised Key Length Versus Surviving Entropy}
\label{sec:phase6}

This is the measurement the paper is built around.  For each instance random
queries are drawn until the count holds steady for eight queries running or a
budget of one hundred and twenty queries is spent, and the surviving entropy
$\log_2|V_T|$ is what is reported.

The campaign covers seventy instances, reaching across all ten ISCAS-85
circuits in the release at advertised key lengths of $32$, $64$, $128$ and
$256$ bits.  Fifty-seven of them plateau.  Of the rest, four use up the query
budget with the count still dropping and nine stop because an engine gave out,
and all thirteen are read as weaker bounds than the others.
Table~\ref{tab:phase6} collects the instances by circuit and key length and
records the weak rows in each group, and Table~\ref{tab:phase6full} lays the
same campaign out one instance at a time.

\begin{table}[t]
\caption{Surviving key entropy on the published Trust-Hub
obfuscation release, summarized by circuit and advertised key length.
Each row aggregates the random-insertion, secure-logic-locking and
cone-size schemes at that key length.  Ranges rather than means,
because the spread across schemes is part of the result.  Every
figure is an upper bound on the surviving secret; the last column
counts rows within the group that had not plateaued when their budget
ran out and are therefore weaker bounds still.  The full
per-instance table is in the extended version.}
\label{tab:phase6}
\centering
\small
\setlength{\tabcolsep}{4pt}
\begin{tabular}{lrrrrr}
\toprule
circuit & $k$ & $n$ & queries & bits lost & soft\\
\midrule
\texttt{c432} & 32 & 3 & 41--120 & 13.0--31.0 & 2\\
\texttt{c432} & 64 & 2 & 9--120 & 11.2--16.9 & 2\\
\midrule
\texttt{c499} & 32 & 3 & 13--27 & 8.0--23.3 & 0\\
\texttt{c499} & 64 & 3 & 15--31 & 21.0--41.0 & 0\\
\texttt{c499} & 128 & 1 & 63 & 24.7 & 0\\
\midrule
\texttt{c880} & 32 & 2 & 33--36 & 21.0--28.0 & 0\\
\texttt{c880} & 64 & 2 & 38--51 & 44.2--55.5 & 0\\
\midrule
\texttt{c1355} & 32 & 3 & 25--54 & 17.3--29.0 & 0\\
\texttt{c1355} & 64 & 2 & 33--73 & 23.2--55.4 & 0\\
\midrule
\texttt{c1908} & 32 & 3 & 18--37 & 8.0--25.0 & 0\\
\texttt{c1908} & 64 & 3 & 16--48 & 4.7--50.0 & 1\\
\texttt{c1908} & 128 & 2 & 1--72 & 8.3--81.0 & 1\\
\midrule
\texttt{c2670} & 32 & 3 & 25--32 & 12.1--17.4 & 0\\
\texttt{c2670} & 64 & 3 & 24--120 & 14.2--42.4 & 1\\
\texttt{c2670} & 128 & 2 & 3--37 & 31.5--63.8 & 1\\
\midrule
\texttt{c3540} & 32 & 3 & 17--46 & 11.0--27.0 & 0\\
\texttt{c3540} & 64 & 3 & 21--39 & 18.7--52.4 & 0\\
\texttt{c3540} & 128 & 2 & 61--72 & 83.4--91.5 & 0\\
\midrule
\texttt{c5315} & 32 & 3 & 23--30 & 16.0--31.0 & 0\\
\texttt{c5315} & 64 & 3 & 24--35 & 23.1--54.0 & 0\\
\texttt{c5315} & 128 & 3 & 1--56 & 8.5--114.1 & 2\\
\texttt{c5315} & 256 & 1 & 31 & 152.1 & 1\\
\midrule
\texttt{c6288} & 32 & 3 & 11--20 & 26.3--32.0 & 0\\
\texttt{c6288} & 64 & 3 & 23--42 & 39.2--64.0 & 0\\
\texttt{c6288} & 128 & 1 & 49 & 107.3 & 0\\
\midrule
\texttt{c7552} & 32 & 3 & 24--44 & 23.1--27.4 & 0\\
\texttt{c7552} & 64 & 3 & 20--30 & 48.1--55.0 & 0\\
\texttt{c7552} & 128 & 2 & 1--2 & 25.0--43.3 & 2\\
\bottomrule
\end{tabular}
\end{table}

\begin{table*}[t]
\caption{Surviving key entropy on published Trust-Hub obfuscation
benchmarks.  RN is random key-gate insertion, SL secure logic
locking, CS logic cone size based.  Every figure is an upper bound on
the surviving secret.  A dagger marks a row whose count was still
falling when the budget of 120 queries ran out; a double dagger a row
on which the engine gave up before that.  Both are weaker bounds than
the rest.  The full list of seventy instances is split into two panels,
read left panel first and then right; each row carries its own
advertised key length.}
\label{tab:phase6full}
\centering
\footnotesize
\setlength{\tabcolsep}{5pt}
\begin{minipage}[t]{0.48\textwidth}
\centering
\begin{tabular}{lrrrr}
\toprule
benchmark & adv.\ $k$ & queries & $\log_2|V_T|$ & lost\\
\midrule
\texttt{c5315-RN2560}$^{\ddagger}$ & 256 & 31 & 103.91 & 152.09\\
\midrule
\texttt{c5315-RN1280} & 128 & 56 & 13.91 & 114.09\\
\texttt{c6288-SL1280} & 128 & 49 & 20.67 & 107.33\\
\texttt{c3540-RN1280} & 128 & 61 & 36.51 & 91.49\\
\texttt{c3540-SL1280} & 128 & 72 & 44.64 & 83.36\\
\texttt{c1908-CS1280} & 128 & 72 & 47.01 & 80.99\\
\texttt{c2670-RN1280} & 128 & 37 & 64.24 & 63.76\\
\texttt{c7552-SL1280}$^{\ddagger}$ & 128 & 2 & 84.72 & 43.28\\
\texttt{c2670-CS1280}$^{\ddagger}$ & 128 & 3 & 96.54 & 31.46\\
\texttt{c5315-SL1280}$^{\ddagger}$ & 128 & 19 & 98.80 & 29.20\\
\texttt{c7552-RN1280}$^{\ddagger}$ & 128 & 1 & 103.03 & 24.97\\
\texttt{c499-CS1280} & 128 & 63 & 103.28 & 24.72\\
\texttt{c5315-CS1280}$^{\ddagger}$ & 128 & 1 & 119.46 & 8.54\\
\texttt{c1908-SL1280}$^{\ddagger}$ & 128 & 1 & 119.65 & 8.35\\
\midrule
\texttt{c6288-RN640} & 64 & 23 & 0.00 & 64.00\\
\texttt{c880-RN640} & 64 & 38 & 8.49 & 55.51\\
\texttt{c1355-RN640} & 64 & 33 & 8.59 & 55.41\\
\texttt{c6288-SL640} & 64 & 35 & 8.59 & 55.41\\
\texttt{c7552-RN640} & 64 & 30 & 9.00 & 55.00\\
\texttt{c5315-CS640} & 64 & 24 & 10.00 & 54.00\\
\texttt{c7552-CS640} & 64 & 26 & 10.00 & 54.00\\
\texttt{c5315-RN640} & 64 & 35 & 11.00 & 53.00\\
\texttt{c3540-RN640} & 64 & 25 & 11.59 & 52.41\\
\texttt{c1908-RN640} & 64 & 37 & 14.00 & 50.00\\
\texttt{c7552-SL640} & 64 & 20 & 15.94 & 48.06\\
\texttt{c3540-CS640} & 64 & 21 & 17.59 & 46.41\\
\texttt{c880-SL640} & 64 & 51 & 19.81 & 44.19\\
\texttt{c2670-RN640} & 64 & 24 & 21.59 & 42.41\\
\texttt{c499-RN640} & 64 & 31 & 23.00 & 41.00\\
\texttt{c6288-CS640} & 64 & 42 & 24.84 & 39.16\\
\texttt{c1908-SL640} & 64 & 48 & 27.84 & 36.16\\
\texttt{c2670-CS640} & 64 & 35 & 32.24 & 31.76\\
\texttt{c1355-SL640} & 64 & 73 & 40.84 & 23.16\\
\texttt{c5315-SL640} & 64 & 33 & 40.86 & 23.14\\
\texttt{c499-CS640} & 64 & 15 & 42.81 & 21.19\\
\bottomrule
\end{tabular}
\end{minipage}\hfill
\begin{minipage}[t]{0.48\textwidth}
\centering
\begin{tabular}{lrrrr}
\toprule
benchmark & adv.\ $k$ & queries & $\log_2|V_T|$ & lost\\
\midrule
\texttt{c499-SL640} & 64 & 17 & 43.00 & 21.00\\
\texttt{c3540-SL640} & 64 & 39 & 45.28 & 18.72\\
\texttt{c432-RN640}$^{\ddagger}$ & 64 & 9 & 47.14 & 16.86\\
\texttt{c2670-SL640}$^{\dagger}$ & 64 & 120 & 49.81 & 14.19\\
\texttt{c432-SL640}$^{\dagger}$ & 64 & 120 & 52.77 & 11.23\\
\texttt{c1908-CS640}$^{\ddagger}$ & 64 & 16 & 59.35 & 4.65\\
\midrule
\texttt{c6288-RN320} & 32 & 11 & 0.00 & 32.00\\
\texttt{c432-SL320} & 32 & 41 & 1.00 & 31.00\\
\texttt{c5315-CS320} & 32 & 23 & 1.00 & 31.00\\
\texttt{c5315-RN320} & 32 & 26 & 1.00 & 31.00\\
\texttt{c6288-SL320} & 32 & 16 & 1.00 & 31.00\\
\texttt{c1355-RN320} & 32 & 54 & 3.00 & 29.00\\
\texttt{c880-RN320} & 32 & 33 & 4.00 & 28.00\\
\texttt{c7552-CS320} & 32 & 24 & 4.58 & 27.41\\
\texttt{c7552-RN320} & 32 & 44 & 4.58 & 27.41\\
\texttt{c3540-RN320} & 32 & 46 & 5.00 & 27.00\\
\texttt{c6288-CS320} & 32 & 20 & 5.70 & 26.30\\
\texttt{c1908-RN320} & 32 & 18 & 7.00 & 25.00\\
\texttt{c3540-SL320} & 32 & 30 & 8.59 & 23.41\\
\texttt{c499-RN320} & 32 & 27 & 8.70 & 23.30\\
\texttt{c7552-SL320} & 32 & 24 & 8.91 & 23.09\\
\texttt{c1355-SL320} & 32 & 44 & 10.69 & 21.31\\
\texttt{c880-SL320} & 32 & 36 & 11.00 & 21.00\\
\texttt{c2670-CS320} & 32 & 29 & 14.57 & 17.43\\
\texttt{c1355-CS320} & 32 & 25 & 14.68 & 17.32\\
\texttt{c2670-RN320} & 32 & 25 & 15.81 & 16.19\\
\texttt{c5315-SL320} & 32 & 30 & 16.00 & 16.00\\
\texttt{c499-SL320} & 32 & 13 & 18.00 & 14.00\\
\texttt{c1908-CS320} & 32 & 37 & 18.63 & 13.37\\
\texttt{c432-CS320}$^{\dagger}$ & 32 & 120 & 18.77 & 13.23\\
\texttt{c432-RN320}$^{\dagger}$ & 32 & 120 & 19.02 & 12.98\\
\texttt{c2670-SL320} & 32 & 32 & 19.92 & 12.08\\
\texttt{c3540-CS320} & 32 & 17 & 21.00 & 11.00\\
\texttt{c1908-SL320} & 32 & 18 & 23.98 & 8.02\\
\texttt{c499-CS320} & 32 & 24 & 24.00 & 8.00\\
\bottomrule
\end{tabular}
\end{minipage}
\end{table*}

\begin{figure}[t]
\centering
\includegraphics[width=0.86\columnwidth]{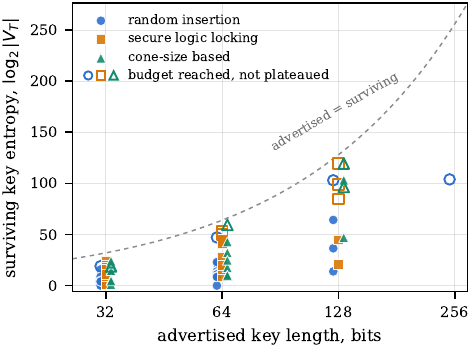}
\caption{Advertised key length against surviving entropy, seventy published
instances over all ten ISCAS-85 circuits.  The dashed curve is
$\log_2|V_T| = k$, which the horizontal axis renders as a curve because it is
logarithmic.  Every instance lies below it, most far below, and two sit on the
axis at zero because their keys are uniquely determined.  Filled markers
plateaued under random queries.  Open markers did not, so their true surviving
entropy is lower than plotted.  The three schemes are offset horizontally
within each key length so that they do not overlap.}
\label{fig:e2surv}
\end{figure}

\begin{figure}[t]
\centering
\includegraphics[width=0.92\columnwidth]{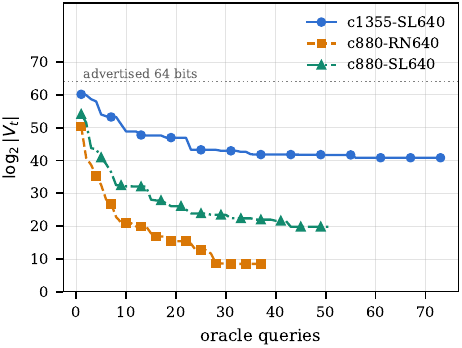}
\caption{Exact version-space size against oracle queries for three 64-bit
instances that plateaued.  The dotted line is the advertised key length.  Each
curve falls below it on the first query and flattens well short of it.}
\label{fig:e2traj}
\end{figure}

Fig.~\ref{fig:e2surv} plots the result and Fig.~\ref{fig:e2traj} shows three
trajectories behind it.  Every instance loses key material, and the loss runs
from $4.65$ to $152.09$ bits.  \bench{c5315-RN2560} advertises a $256$-bit key and $31$ random queries
remove $152$ bits of it; \bench{c5315-RN1280} advertises $128$ and retains
$13.91$; \bench{c432-SL320} advertises $32$ and retains one bit, which is two
keys.  Two instances retain nothing at all: \bench{c6288-RN320} is uniquely
determined by eleven random queries and \bench{c6288-RN640} by twenty-three.
In both cases $|V_T| = 1$ and the advertised $32$ and $64$ bit keys are
fully recovered by uniform random querying, without a solver and without a
chosen input pattern.

That the circuit is c6288 in both cases is worth a comment, because c6288 is
the canonical circuit whose decision diagram cannot be built: it is a
$16\times16$ array multiplier, and the exponential lower bound on any ordered
diagram for a middle output bit of an integer multiplier is a
theorem~\cite{bryant1991} holding for every variable order, so one would
expect Engine~B to fail here first.  It does not, because the two diagrams are
not the same object.  Bryant's bound is on a diagram of the circuit's
\emph{function} over the primary inputs, whereas Engine~B fixes the primary
inputs to a query pattern and builds a diagram over the \emph{key} variables
whose satisfying assignments are the surviving keys.  Nothing in the
multiplier bound constrains that, and in practice Engine~B carries every
locked c6288 variant it was given with an accumulated diagram running from
$33$ nodes to $3{,}593$ over the seven of them, while Engine~A fails on all of
them.

Two broader patterns are visible in Table~\ref{tab:phase6}.  The loss grows
with the advertised key length rather than shrinking, so a longer key buys
more surviving entropy in absolute terms while the fraction lost does not
fall, and at $128$ bits the losses are the largest in the campaign.  And the
circuit matters more than the scheme: c880, c1355, c5315 and c6288 lose most
of their key material under random insertion and secure logic locking alike,
while c432 holds far more under both.  That is a statement about circuit
structure rather than locking method, and it is the kind a count can make and
an iteration count cannot.

\paragraph{These results are upper bounds}  Three qualifications bear on the
figures, all pushing the same way and none of them working in the advertised
key's favour: the queries are drawn uniformly at random and an attacker who
chose them would get there faster, the stopping rule is a plateau under random
queries and not a proof that no further query splits the survivors, and any
further query can only shrink $V$.  So $\log_2|V_T|$ is an upper bound on the
surviving secret, and the campaign supports ``at most 8.49 bits remain'' where
it does not support ``exactly 8.49''.  On the thirteen rows that never
plateaued the bound is looser still.  The exception is the two instances at
$|V_T| = 1$, where a version space of a single key has nothing left to lose.

\subsection{Reading the Campaign by Circuit Family}
\label{sec:byfamily}

Table~\ref{tab:phase6full} repays a reading by circuit rather than by scheme,
because it is the circuit, more than the construction placed on it, that
decides how much of the key survives.  That is itself one of the campaign's
findings, and it inverts the usual order of the defender's attention: a lock
chosen without regard to the circuit it protects has been chosen on the less
important of the two variables.  What follows walks the circuits from the one
that gives up everything to the one that holds the most, so that the range the
grouped table compresses into intervals can be seen instance by instance.

The array multiplier c6288 is the circuit that gives up everything.  Both of
the uniquely-determined instances in the campaign are locked c6288 circuits:
\bench{c6288-RN320} collapses to a single surviving key after eleven random
queries and \bench{c6288-RN640} after twenty-three, and its $128$-bit instance
\bench{c6288-SL1280} sheds $107.33$ bits, the largest single loss of any
$128$-bit row not on c5315.  That the most arithmetically complex circuit in
the suite is the easiest to strip is the opposite of what a decision-diagram
intuition would predict, and Section~\ref{sec:floor} explains why the diagram
bound that would forbid it does not touch the object Engine~B actually builds.
Its arithmetic density, which makes the multiplier hard to represent as a
function of its inputs, does nothing to keep the key secret once the inputs are
fixed.

The wide, output-heavy circuits sit next to it.  On c5315 the losses are the
largest in the whole campaign: \bench{c5315-RN2560} advertises a $256$-bit key
and thirty-one random queries carry off $152.09$ bits of it, and
\bench{c5315-RN1280} retains only $13.91$ of an advertised $128$.  On c880 a
$64$-bit key is worth little: \bench{c880-RN640} keeps $8.49$ bits and
\bench{c880-SL640} keeps $19.81$, so secure logic locking buys roughly eleven
bits over random insertion on the same circuit and key length and still leaves
most of the key gone.  c1355, c7552 and the mid-size circuits c1908, c2670 and
c3540 fall in the same band, losing the bulk of a $64$-bit key under random
queries alone, with c1355 dropping to $8.59$ bits at $64$ and to $3.00$ at
$32$.  Across all of these the scheme moves the number by a handful of bits
while the circuit moves it by tens.

c432 is the circuit that holds the most, and it is also where the counting is
hardest, which is not a coincidence.  Its $64$-bit rows retain more than any
other circuit's, and its counts are the slowest to settle: the instances that
run out of query budget with the count still falling are concentrated on c432,
so \bench{c432-SL640} at $52.77$ bits and \bench{c432-CS320} at $18.77$ are
upper bounds looser than the rest, and Section~\ref{sec:floor} re-runs two of
them to show the plateau lies well below where the short budget left it.  The
one c432 instance the method cannot count at all, \bench{c432-RN640}, is also
here, for the same underlying reason: a version space that stays large is both
a strong lock and an expensive diagram.  The exception within the family is
\bench{c432-SL320}, which secure logic locking strips to a single surviving
bit, which is two keys, and which shows that even the most resistant circuit is
not uniformly so.

Two patterns hold across the walk.  Reading down the key-length column the
absolute loss grows rather than shrinks, so a longer key buys more surviving
entropy in absolute terms while the fraction lost does not fall and the
campaign's largest losses sit at its longest keys.  Reading across schemes at a
fixed circuit and key length the spread is narrow beside the spread across
circuits, which is the quantitative form of the finding that the circuit
matters more than the scheme.  The per-instance table is included so that both
patterns can be checked row by row rather than taken from the grouped summary
of Table~\ref{tab:phase6} on trust, and the grouped table reports intervals
precisely because the spread inside a group is part of what the measurement
found.

\subsection{The BDD-Based Family}
\label{sec:bdd}

The release carries a fourth family built by manipulating decision diagrams
rather than by inserting key gates.  It has three variants, written BE, BR and
BS in the instance names, covering twenty-four instances over six circuits at
$28$, $32$, $64$ and $128$ key bits.  We report it separately from
Table~\ref{tab:phase6} because the variants behave so differently under the
same measurement that any average over them would describe none of them.

\begin{table}[t]
\caption{The BDD-based obfuscation family.  The width and diagram
columns are the screen at three queries; the last three are the
campaign run to a plateau.  Both engines reach every instance but
one.  The three variants separate: BE loses most of the key, BR part
of it, and BS none of it.}
\label{tab:bdd}
\centering
\small
\setlength{\tabcolsep}{3.5pt}
\begin{tabular}{@{}lrrrrrr@{}}
\toprule
instance & bits & width & diagram & $t$ & $\log_2|V_t|$ & lost\\
\midrule
\bench{c1908-BE320} & 32 & 7 & 7,557 & 62 & 4.58 & 27.41\\
\bench{c3540-BE320} & 32 & 6 & 650 & 53 & 1.58 & 30.41\\
\bench{c432-BE280} & 28 & 6 & 410 & 53 & 3.00 & 25.00\\
\bench{c5315-BE320} & 32 & 6 & 1,049 & 49 & 2.00 & 30.00\\
\bench{c7552-BE1280} & 128 & 7 & 195 & 121 & 2.00 & 126.00\\
\bench{c7552-BE320} & 32 & 7 & 43 & 54 & 0.00 & 32.00\\
\bench{c7552-BE640} & 64 & 7 & 105 & 48 & 1.58 & 62.41\\
\bench{c880-BE320} & 32 & 7 & 2,533 & 57 & 2.00 & 30.00\\
\midrule
\bench{c1908-BR320} & 32 & 3 & 1,150 & 86 & 29.03 & 2.97\\
\bench{c3540-BR320} & 32 & 6 & 980 & 326 & 17.41 & 14.59\\
\bench{c432-BR320} & 32 & 17 & 0 & 41 & 32.00 & 0.00\\
\bench{c5315-BR320} & 32 & 6 & 40 & 98 & 24.66 & 7.34\\
\bench{c7552-BR1280} & 128 & 5 & 137 & 388 & 79.20 & 48.80\\
\bench{c7552-BR320} & 32 & 6 & 37 & 115 & 19.13 & 12.87\\
\bench{c7552-BR640} & 64 & 6 & 70 & 148 & 40.00 & 24.00\\
\bench{c880-BR320} & 32 & 5 & 2,998 & 157 & 21.08 & 10.92\\
\midrule
\bench{c1908-BS320} & 32 & 5 & 0 & 41 & 32.00 & 0.00\\
\bench{c3540-BS320} & 32 & 6 & 0 & 41 & 32.00 & 0.00\\
\bench{c432-BS320} & 32 & 17 & 0 & 41 & 32.00 & 0.00\\
\bench{c5315-BS320} & 32 & 5 & 0 & 41 & 32.00 & 0.00\\
\bench{c7552-BS1280} & 128 & $>25$ & 1 & 41 & 128.00 & 0.00\\
\bench{c7552-BS320} & 32 & 6 & 1 & 41 & 32.00 & 0.00\\
\bench{c7552-BS640} & 64 & 6 & 1 & 41 & 64.00 & 0.00\\
\bench{c880-BS320} & 32 & 6 & 0 & 41 & 32.00 & 0.00\\
\bottomrule
\end{tabular}
\end{table}

Table~\ref{tab:bdd} reports the whole family.  The first two numeric columns
are the screen at three queries: Engine B reaches every instance and Engine A
reaches twenty-three of the twenty-four, at widths running from $3$ to $17$,
which is below anything the other families reach at comparable key lengths.
That makes this family the counterexample to any reading of
Section~\ref{sec:engines} in which Engine A is simply the weaker of the two.
The gate counts, which run from $5{,}774$ to $190{,}305$, are an order of
magnitude above the rest of the release because these netlists are emitted as
flattened Boolean assignments rather than gate instances, and the reader
counts one operation per gate.

The last three columns are the campaign, and all twenty-four plateau.  The BE
variants lose from $25.00$ to $126.00$ bits, with \bench{c7552-BE320} losing
its whole $32$-bit key.  The BR variants lose from $2.97$ to $48.80$ bits, and
the amount varies by circuit rather than by key length.  The BS variants lose
nothing at all: the count stands at $2^{k}$ after every query drawn, at
$k=32$, $64$ and $128$ alike, and the accumulated diagram never exceeds one
node.

One row needs an explanation that has nothing to do with locking.
\bench{c432-BR320} sits with the BS variants rather than with the other BR
ones because the release ships it and \bench{c432-BS320} as byte-identical
files; hashing every unsynthesized netlist gives $294$ distinct files across
$295$ instances, and that pair is the only duplicate, so the BR variant has
seven independent instances here and not eight.

A zero deserves a check, and it survives four.  The key ports are read in the
body of the netlist rather than merely declared; they reach primary outputs
along a path through the gate graph; the reader agrees with an independent
simulator on these instances, so this is not a parse artifact; and a
satisfiability query over two copies of the netlist that share the primary
inputs and carry independent keys is satisfiable, which exhibits an input on
which two keys disagree.  The key is therefore not inert.  What the zero says
is that the set of input patterns on which the key changes an output is small
enough that random draws do not find one, which is a property the
surviving-entropy measurement is built to see and an iteration count is not.
It is also the sharpest illustration in this paper of why every number
reported so far is an upper bound, and Section~\ref{sec:adversarial} measures
the same instances under queries that are chosen rather than drawn.

\subsection{The Floor and Bound Tightness}
\label{sec:floor}

An upper bound is only as useful as one's grasp of what it bounds, and here the
answer is exact.  Write
\[
  E(K^*) \;=\; \{\,K : C(x,K) = C(x,K^*) \text{ for every input } x\,\}
\]
for the set of keys functionally equivalent to the correct one $K^*$.

\begin{proposition}\label{prop:floor}
$|V_t|$ is non-increasing in $t$ and bounded below by $|E(K^*)| \geq 1$.  It is
therefore eventually constant, and under queries drawn uniformly at random it
converges almost surely to exactly $|E(K^*)|$.
\end{proposition}

The proof is a line.  A query only ever removes keys; no query can remove a key
of $E(K^*)$, because such a key matches the oracle on every input; and a
non-increasing sequence of positive integers must settle.  So every instance
plateaus, and the plateau is not a creature of the stopping rule but the size of
the equivalent-key class, which is exactly what remains of the secret once no
observation can separate what is left.  It also meets an objection the
definitional line raises~\cite{beerel2022}, that a count of keys is not a count
of distinguishable functions: the two differ by precisely the equivalent-key
class, so quoting the floor beside the count is what holds them apart.  Two
consequences bear on how Table~\ref{tab:phase6} should be read.

\paragraph{A plateau that has not appeared is a budget, not a barrier}  Four
instances exhausted the query budget with their counts still falling, three of
them on c432, and their tails are very flat: over their last forty queries
\bench{c432-CS320} fell $0.056$ bits and \bench{c432-RN320} fell $0.270$, so a
rule that stops after eight unchanged counts can essentially never fire, since
a small drop keeps landing inside every window.  We re-ran five instances with
a budget of three thousand queries and a stricter plateau rule.  Four of the
five then plateau, two of them far below where the shorter budget left them:
\bench{c432-RN320} reaches $11.03$ bits after $630$ queries against the
$19.02$ reported at $120$, and \bench{c432-CS320} reaches $3.00$ after
$1{,}233$ against $18.77$, nearly sixteen further bits on an instance the
shorter budget made look resilient.  The limit was the budget and not the
structure.

\paragraph{The bound can be certified tight}  $V_t = E(K^*)$ exactly when no
input pattern separates two surviving keys, and for one output that is
decidable by quantifying the key variables out of the joint diagram.  With
$A_o(x) = \exists K\,[\,f_o(x,K) \wedge V_t(K)\,]$ and
$B_o(x) = \exists K\,[\,\neg f_o(x,K) \wedge V_t(K)\,]$, a distinguishing
pattern exists precisely when $A_o \wedge B_o$ is satisfiable for some $o$;
when it is not, $\log_2|V_t|$ is exact rather than an upper bound.  That is
the same condition a satisfiability-based attack tests when it looks for its
next distinguishing pattern, so the attack's termination condition and the
tightness certificate are one statement.  We check it against exhaustive
search on generated instances, where the version space can be enumerated and
every input pattern tried, and it agrees in every case.  Where $|V_t| = 1$ it
is immediate, since one key cannot be separated from itself, and that is how
the two uniquely-determined c6288 rows are certified: their $32$ and $64$ bit
keys are recovered \emph{exactly} rather than to within an upper bound.

\paragraph{This form of the certificate has a barrier the counting does not}
As formulated above the check needs the joint diagram over inputs and keys,
which is a diagram of the circuit's function, so the multiplier bound
of~\cite{bryant1991} applies to it even though Section~\ref{sec:phase6} showed
it does not apply to Engine~B.  The consequence is stronger than a budget: on
c6288 the joint diagram cannot be built at any budget for any variable order,
so had those two instances retained even two keys they would have been
permanently uncertifiable this way, and the certificates we do have there rest
entirely on the one-key rule.  The barrier is specific to the diagram
formulation, and Section~\ref{sec:adversarial} certifies through a solver
instead, which removes it outright.

\paragraph{Failure cases}  We report one instance,
\bench{c432-RN640}, as a failure rather than dropping it.  At a node budget of
eight million the engine gives up after two queries, and raising the budget to
a hundred million buys seven more before fifteen minutes of wall time run out,
with the accumulated diagram at 44.6 million nodes and the count still above
$2^{47}$.  This is the failure mode Section~\ref{sec:engineb} implies: the
diagram represents the version space, so it is large precisely when the
version space is large, which is the regime in which a scheme is doing its
job.  The engine fails where the answer would be least welcome to the
attacker, so a scheme that survives here has survived measurement rather than
been shown secure.  It is the one direction in which the method cannot be
trusted to be conservative.

\subsection{Chosen Queries Versus Random Queries}
\label{sec:adversarial}

Every entropy figure so far rests on queries drawn uniformly at random, which
is the reason each one is stated as an upper bound.  What separates that bound
from what a chooser of queries would leave behind has until now been an
assumption; this section puts a number on it.

The tool is the satisfiability question of Section~\ref{sec:floor}, turned from
a termination test into a generator of queries.  Two copies of the netlist share
the primary inputs of a candidate query and carry independent key vectors, every
query already asked is asserted of both, and the objective demands that the two
copies differ on some primary output.  A model of that formula is an input
pattern on which two surviving keys disagree, which is exactly the distinguishing
input a satisfiability-based attack would compute; it becomes the next query, and
the loop turns again.  When the formula is unsatisfiable no input separates two
surviving keys, so $V_t = E(K^*)$ and the entropy reported for that instance is
exact rather than bounded.

\begin{table}[t]
\caption{Random queries against solver-chosen queries.  The two
middle columns are measured at the same query count $t$, so the
only difference between them is how the query was chosen.  The last
column is the same instance under the campaign's own budget of
100 random queries.  A dagger marks an instance where the solver
reported that no distinguishing input remains, which is a proof
that the surviving set is the equivalence class of the key.}
\label{tab:adversarial}
\centering
\small
\setlength{\tabcolsep}{4pt}
\begin{tabular}{lrrrrr}
\toprule
instance & key bits & $t$ & random & chosen & random, 100\\
\midrule
\bench{c1355-RN320} & 32 & 12 & 10.60 & 5.46 & 1.00\\
\bench{c1908-BE320} & 32 & 12 & 9.91 & 1.00 & 2.00\\
\bench{c1908-BS320} & 32 & 1 & 32.00 & 31.91 & 32.00\\
\bench{c432-BE280} & 28 & 8$^\dagger$ & 4.17 & 1.00 & 1.00\\
\bench{c432-BR320} & 32 & 12 & 32.00 & 31.91 & 32.00\\
\bench{c432-BS320} & 32 & 12 & 32.00 & 31.91 & 32.00\\
\bench{c432-RN320} & 32 & 12 & 14.15 & 10.58 & 2.58\\
\bench{c432-SL320} & 32 & 11$^\dagger$ & 13.22 & 0.00 & 1.00\\
\bench{c5315-BE320} & 32 & 7$^\dagger$ & 1.00 & 1.00 & 1.00\\
\bench{c5315-BS320} & 32 & 10$^\dagger$ & 32.00 & 24.44 & 32.00\\
\bench{c880-BE320} & 32 & 9$^\dagger$ & 7.49 & 5.58 & 6.17\\
\bench{c880-BS320} & 32 & 12 & 32.00 & 19.61 & 32.00\\
\bench{c880-RN320} & 32 & 11$^\dagger$ & 3.00 & 2.00 & 2.00\\
\bottomrule
\end{tabular}
\end{table}

Two columns of Table~\ref{tab:adversarial} are measured at the same query
count, so the only difference between them is how the query was chosen.  The
chosen column is at or below the random column on every instance measured.
Per query, choosing is worth something everywhere, and on \bench{c432-SL320}
eleven chosen queries determine the key exactly while eleven random ones leave
$13.22$ bits.

The last column changes the reading.  It is the same instance under a hundred
random queries rather than twelve, and on the key-gate families it falls below
the chosen column: a hundred random queries leave $1.00$ bit of
\bench{c1355-RN320} where twelve chosen ones leave $5.46$.  Random queries are
free and chosen ones are not, so the two columns answer different questions.
Per query, choosing wins; per unit of effort, on an instance where random
patterns separate surviving keys at all, drawing more of them wins.  That
leaves the instances where they do not.  The BS variants of
Section~\ref{sec:bdd} stand at the full key space after a hundred random
queries, and the last column does not overtake the chosen column at any
budget, because the quantity it measures does not move.

The summary is therefore not that one way of choosing a query dominates.  The
two are separated by the same property of the lock that this paper measures.
Where a random input pattern separates surviving keys often, the random-query
bound is close and cheap to obtain; where the construction makes that rare,
the bound is far above the truth, no number of random queries closes it, and
the measurement gives no warning of that on its own.  The certificate does.

\paragraph{The certificate, through a solver}  Six instances in
Table~\ref{tab:adversarial} carry a proof that the surviving set is the
equivalence class of the key.  That is the condition
Section~\ref{sec:floor} states, decided without building a diagram over the
inputs and the keys, so the multiplier bound of~\cite{bryant1991} does not
apply and the barrier that section leaves open is closed.  The sharpest case
is \bench{c5315-BS320}: a hundred random queries leave the full $32$-bit key
space standing.  Ten chosen queries reach $24.44$ bits, and the solver then
reports that no distinguishing input remains, so that figure is exact rather
than a bound.  The key on that instance is worth $7.56$ bits, and random
querying recovers none of them.

\paragraph{Chosen queries are not free}  The cost falls
in the same place as the benefit.  Producing a chosen query means finding a
satisfying assignment for a formula asserting that two surviving keys disagree
somewhere, and a construction that makes such inputs rare is a construction
that makes that formula hard.  On \bench{c5315-BS320} the ten queries cost
three seconds of solver time in total; on \bench{c880-BS320} twelve queries
cost six minutes; on \bench{c1908-BS320} the first query alone cost
thirty-five minutes and the budget allowed no second one.  The advantage of
choosing is real on these locks, and it is bought with solver time rather than
with query count.

\subsection{Engines and Implementation}
\label{sec:impl}

The two engines that compute $|V_t|$ are independent implementations of the two
representations discussed above, and they are kept independent on purpose, so
that agreement between them is evidence rather than a restatement of one
computation.  They are described here in enough detail that either could be
reconstructed, because the reachability of an instance is a property of the
engine as much as of the lock, and a reader weighing a surviving-entropy figure
should know which engine produced it and where that engine breaks.

\paragraph{Engine~A: elimination over the residual factor graph}
Engine~A is the sum--product engine of Section~\ref{sec:cec}, run over the
residual gate-level factor graph rather than over the Haar block tree.  It
begins each query by fixing the primary inputs to the query pattern and
propagating constants forward: a gate all of whose inputs are now determined
becomes a constant and is removed, and a gate with a single surviving symbolic
input collapses to a buffer or an inverter.  What is left is a residual circuit
over the key bits and the internal nets that still depend on them, which for
most designs is a small fraction of the netlist because most of a circuit does
not touch the key.  Each surviving gate contributes the indicator factor
of~\eqref{eq:gatefactor}, the pinned outputs contribute unit factors, and the
count of satisfying assignments of the resulting factor set is $|V_t|$.

Gates of fan-in greater than two are decomposed into two-input primitives
before the graph is built, because the elimination cost is read off a graph
whose factors have bounded scope and a wide gate would otherwise hide a large
factor inside a single vertex; the decomposition is checked against brute force
on four hundred cases so that it cannot silently change the function it
represents.  The elimination order is produced by the greedy minimum-fill
heuristic, which repeatedly removes the variable whose elimination adds the
fewest fill edges to the remaining graph.  We use a heuristic rather than an
optimal order because computing the true treewidth is itself intractable, and a
heuristic order gives an upper bound on the induced width, so the widths the
engine reports are conservative: the true factor width is no larger than the
number printed.  That number is the induced width of the order actually used,
which is the factor width of Section~\ref{sec:twowidths}, and the cost of the
elimination is singly exponential in it.  The engine records the width it
achieved on every instance and every query, so that an instance the engine
cannot finish is reported as a width that crossed the threshold rather than as
an unexplained timeout, and a screen can decline an instance in advance by
running only the ordering and reading the width off it without computing a
count.

\paragraph{Engine~B: a diagram over the key variables}
Engine~B represents the version space directly.  It simulates the netlist
symbolically with the primary inputs of a query fixed, so that every net
becomes a Boolean function of the key bits alone; for each observed output it
conjoins the agreement between that net's function and the observed value, and
it conjoins these constraints across queries into one reduced ordered binary
decision diagram whose satisfying assignments are exactly $V_t$.  The count is
then one traversal of that diagram.  Because the diagram is over the key
variables only, its variable order ranges over a few dozen to a few hundred
bits rather than over the primary inputs, and dynamic reordering keeps that
order workable as constraints accumulate; the diagram is an object that
represents the surviving set and so it shrinks, not grows, as queries pin the
key down.

Engine~B exists in two forms that share neither source code nor data structure:
a self-contained Python implementation and a C implementation built on the CUDD
package~\cite{cudd}.  The C form counts with the arbitrary-precision routine
\texttt{Cudd\_ApaCountMinterm} rather than the double-precision
\texttt{Cudd\_CountMinterm}, because a $64$-bit key space is $2^{64}$ and a
minterm count returned as a double-precision float is already wrong above
$2^{53}$, well inside the range of interest; the count the paper reports is
therefore exact to the last digit and not a floating-point approximation.  The
two forms were written from the specification independently rather than one
ported from the other, so that a disagreement between them would point to a
defect, and during development that discipline did find defects that a single
implementation would have carried into the results.

\paragraph{Why two engines}
The two engines pay in different currencies, and that is the reason for keeping
both.  Engine~A pays in elimination width, which grows as each query adds
another copy of the residual circuit to the factor graph; Engine~B pays in
diagram size, which falls as each query removes keys from the set the diagram
represents.  They therefore fail in opposite regimes: Engine~A on the
many-query, wide-residual instances such as point-function locks, and Engine~B
on the instances whose version space is still enormous, which is exactly where
a lock is doing its job.  Where both finish they agree exactly, and together
with the brute-force enumeration available up to $k=20$ they give a four-way
check on every instance small enough to admit more than one method.  The pair
extends the reachable range in both directions at once, and
Section~\ref{sec:engines} exhibits a single benchmark, \bench{c432-RN640}, on
which the crossover is visible within the first few queries.

\subsection{Reproducibility and the Three-Machine Campaign}
\label{sec:repro}

Every number the key-counting application reports is tied to the
machine-readable record that produced it by a checking gate that carries, at
the time of writing, sixty-one assertions; each prose figure is matched against
the result file it came from, and a figure that is edited without updating its
source, or the reverse, fails the gate.  Figures and tables are regenerated
from the same result files rather than transcribed.  The counting kernel
carries its own suite, described in Section~\ref{sec:cec} for the spectral
case and extended here to key recovery, and it runs on every invocation so that
a released artifact that does not reproduce is caught before it is shipped.

The published campaign was executed across three machines, an Apple-silicon
macOS system and two Linux systems, with the work partitioned so that no
reported quantity mixes hardware in a way that would make it uninterpretable.
The counts, widths and version-space trajectories are exact and machine
independent, so they may be computed anywhere; the wall-clock measurements, by
contrast, are confined to a single machine, because a budget expressed in
seconds buys a different number of queries on different hardware and a table
that mixed two machines' times would not be a table.  The partition, the merge
step that reassembles the shards, and the gate that re-checks every prose
number after the merge are all part of the released artifact.

\subsection{A Unified View: One Method, Two Applications}
\label{sec:unified}

The two applications of this paper are the same computation viewed twice.  In
both, a hidden object is constrained by observations, the surviving candidates
form a version space, and the quantity of interest is its exact size.  In the
equivalence-checking application the object is a Boolean function, the
observations are modified-Haar coefficients, and the constraints live on the
dyadic block tree; in the locking application the object is a key, the
observations are oracle responses, and the constraints live on the gate-level
factor graph.  The counting engine is identical: a sum--product recursion over
a tree-structured or gate-structured factor graph, exact over the counting
semiring, singly rather than doubly exponential in the width of the graph
because the auxiliary variables it carries are determined by the ones being
counted rather than existentially quantified away.

The correspondence is more than an analogy, and it is useful in both
directions.  The spectral application is the one in which every quantity has a
closed form and can be checked against exhaustive enumeration, so it is where
the recursion, the ordering rule, and the exactness of the count are validated
against ground truth.  The locking application is the one in which no ground
truth is available, because the version spaces are astronomically large and
the point is precisely that they cannot be enumerated; it inherits the
confidence established in the spectral case and adds the two independent
engines as a cross-check where brute force runs out.  The independence-gap
result of Section~\ref{sec:cec} is the sharpest illustration of why the exact
count is worth the trouble in both settings: a practitioner who cannot compute
$|V_t|$ and multiplies per-observation probabilities instead is wrong by a
lattice index that grows with the coupling among the observations, always in
the direction that overstates how much of the secret survives.  In the
spectral case that error is computed exactly; in the locking case it is the
reason the advertised key length, which is exactly such a product estimate,
stands so far above the surviving entropy this paper measures.

\FloatBarrier
\section{Discussion}
\label{sec:disc}

\subsection{Status of the Equivalence-Checking Application}
\label{sec:cecstatus}

The counting problem that the 2002 method left open (Section~\ref{sec:cec}) is
now closed: the general recursion of Theorem~\ref{thm:dp} evaluates $\kc(C)$
for any conditioned set in time polynomial in the truth-table size, the two
special cases the 2002 work could reach are the extremes of the closed form of
Theorem~\ref{thm:anc}, and the error of the independence approximation it
otherwise resorted to is the computable lattice index of
Proposition~\ref{prop:gap}.  We stress what this does and does not claim.
Combinational equivalence checking as an engineering task has been settled by
satisfiability sweeping and by computer algebra, and a probabilistic checker
competes with neither; the contribution here is the completion of the
underlying counting theory and the exactly-verifiable setting it provides, not
a practical checker.  The value of the CEC application in this paper is
twofold: it finishes a problem that stood open for two decades, and it
validates against closed-form and brute-force ground truth the same machinery
that the locking application must run where no ground truth is available.

\subsection{Limitations of the Locking Application}
\label{sec:notdo}

It is worth being explicit about the boundaries, because a measurement is
easily mistaken for an attack.  Nothing here recovers a key faster than an
existing method, and nothing here is a new locking scheme.  What the method
adds is the ability to say, after each query, how many keys remain, and to say
in advance what a candidate query would be worth.
A distinguishing input is still found by a solver, and a defender who learns
that the version space is small has learned that the scheme is weak, not how
to repair it.  Nor is the machinery new in itself: the inference is the
generalized distributive law, the property that makes it applicable to a
netlist has been known in the counting literature for a decade, and computing
residual uncertainty by counting is routine in quantitative information flow.
What is new is that these have not been put together, and that doing so turns
a quantity the field has called out of reach into one it can report.

\subsection{The Uniform Prior over Keys}
\label{sec:priors}

Counts become posteriors under the uniform prior on the surviving version
space, and for key recovery that prior is more than a convenience: the secret
key is chosen by the designer and is uncorrelated with the netlist an attacker
holds, so before any query every key really is equally likely.  That is a
stronger position than the uniform prior occupies in most counting problems,
where it is a null model adopted for want of anything better.  One
qualification applies.  A key preprocessor that derives an internal key from a
chip-unique value breaks the correspondence between the key an attacker
recovers and the key that unlocks another die; the counting is unaffected,
since it concerns the internal key, but the security interpretation changes
and we do not pursue it here.

\subsection{Scope and What Remains}
\label{sec:scope}

The approach applies when observations are low-rate, meaning that a single
observation leaves many candidates consistent with it, and when the constraint
system admits an elimination order of small induced width.  Oracle queries
against a locked netlist satisfy the first comfortably, since one input
pattern constrains one output response and for a lock of any strength that
leaves a large key set.  Section~\ref{sec:expt} measures the second
directly.

Four things bound what we have shown.  The widths of Section~\ref{sec:e11} are
the widths of the instances we measured, and those are generated adders,
multipliers and comparators under reimplemented locking schemes; the published
release contributes counts rather than widths, so the width claim itself rests
on generated instances.  Under point-function locking the factor width grows
with the query count and the elimination engine stops early, which is inherent
because every query in such a scheme touches every key bit, and it means that
engine is weakest exactly where the counting would be most interesting; the
diagram engine carries the same instances to three hundred queries, so this
bounds one engine rather than the method.  Routing-network locking, which
Section~\ref{sec:lockfail} predicts will defeat the approach, remains
unmeasured, as do cyclic locking and the synthesized variants, for the reasons
given in Section~\ref{sec:screen}.  And while Section~\ref{sec:adversarial}
compares random queries against the patterns a satisfiability-based attack
selects, it does not order a fixed set of queries, so whether ordering by
expected information gain improves on either choice is still open.
Sequential locking is outside the scope of this paper.

\subsection{Extensions}
\label{sec:ext}

Three directions look worthwhile.  A sequential locked design unrolled to a
fixed depth is again a combinational factor graph with state variables shared
between frames, so the machinery applies unchanged and the open question is
how the width grows with unrolling depth.  For approximate keys, AppSAT works
with keys that are wrong on a small fraction of inputs~\cite{shamsi2017}, and
exact counting gives the size of the set of $\varepsilon$-approximate keys
directly by weighting rather than by sampling.  Most interesting to us is the
implication for scheme design: if it is the width of the factor graph that
makes counting affordable, then a construction that wants to be hard to
measure should raise that width rather than the width over key variables.
That is a design criterion nobody currently uses, and routing-network locking
may already satisfy it by accident.

\section{Conclusion}
\label{sec:concl}

This paper has treated two circuit-analysis problems as one, through the
version space: the set of hidden objects still consistent with a set of
observations, counted exactly and reported on a logarithmic scale.  In the
first application the objects are Boolean functions and the observations are
modified-Haar coefficients, and the version-space count is the confidence of a
probabilistic equivalence check; the general recursion, the closed forms, and
the exact independence-gap result close the counting problem that the 2002
method posed and left open.  In the second application the objects are secret
keys and the observations are oracle responses, and the same recursion, run
over the gate-level factor graph, measures what a locked netlist actually
hides.  The rest of this section summarizes the second application, whose
numbers are the paper's empirical contribution.

Key length is an upper bound of the secret rather than a measurement of it,
iteration count indicates how long an attack ran rather than how well it did, and
runtime indicates how hard the solver worked.  None of these three metrics provide the number of
keys an attacker can rule out, and that number is computable.  This
paper computes it exactly for locked netlists of realistic size.

The key counts demonstrate three locking schemes behaving in three quite different ways,
only one of which is well described by its advertised key length.  A
point-function lock removes exactly one key per query, which is what it is
designed to do.  An interference-maximizing lock resists sensitization and then
falls to eight random queries.  A random lock stalls at $8{,}192$ surviving
keys that no further query can separate, so thirteen of its twenty-four key
bits hold nothing at all.  For the obfuscation benchmarks, the same measurement
covers seventy instances across all ten ISCAS-85 circuits and finds that every
one of them retains less security than iadvertised, while a fourth family of
twenty-four other examples divides so sharply that one of its variants
loses nothing whatsoever to random querying even though a solver separates its
keys in under a second.

This last case is the one we emphasize, because it
illustrates where the instrument and intuition part company.  A plateau under
random querying is not by itself a proof of anything; a floor exists below
which no query set can drive the count, and it requires a certificate to indicate
whether a given plateau has been reached.  The suggestion that follows is a
modest one: report $|V_t|$ alongside iteration count when a locking scheme is
evaluated, and determine which locked circuits keep their advertised key length
once their equivalent keys are removed.

\section*{Code and Data Availability}
The companion code is released under the MIT license, and the archived release
carries a DOI that resolves to the most recent archived version.
\begin{center}
\url{https://github.com/mitch-thornton/locked-logic-key-counting}\\[2pt]
\url{https://doi.org/10.5281/zenodo.22218068}
\end{center}
The repository holds both counting engines, the reachability screen, the
plateau certificate, the campaign drivers, and a self-contained structural
Verilog reader, so that no external netlist front end is needed to read the
published benchmarks.  A single command runs a self-check on a fresh clone,
verifying every closed form against exhaustive enumeration and running the
engine parity gates.  The Trust-Hub obfuscation benchmarks are third party and
are not redistributed in the repository, which instead documents where to
obtain them and carries digests for confirming that a download matches the
archives used here.

\section*{Acknowledgment}
This work was not supported by any externally funded projects. An AI agent
was used for grammatical, editing and software development assistance.  
The author is responsible for all results and their description.

\bibliographystyle{IEEEtran}
\bibliography{j2_refs}

\end{document}